\documentclass[10pt,lettersize,journal]{IEEEtran}
\IEEEoverridecommandlockouts
\usepackage{cite}
\usepackage{amsmath,amssymb,amsfonts}
\usepackage{algorithmic}
\usepackage{graphicx}
\usepackage{textcomp}
\usepackage{xcolor}
\usepackage{units}
\usepackage{multirow}
\usepackage{colortbl}
\usepackage{soul}
\usepackage[left=0.625in, right=0.625in,top=0.75in,bottom=1in]{geometry}

\usepackage{amsthm}
\usepackage{thmtools, thm-restate}
\usepackage{algorithm}

\newcommand{\Dc}{\mathcal{D}}

\newcommand{\Eb}{\mathbb{E}}
\newcommand{\e}{\mathbf{e}}

\newcommand{\I}{\mathbf{I}}
\newcommand{\Ic}{\mathcal{I}}

\newcommand{\Mc}{\mathcal{M}}

\newcommand{\Nc}{\mathcal{N}}
\newcommand{\n}{\mathbf{n}}

\newcommand{\w}{\mathbf{w}}

\newcommand{\x}{\mathbf{x}}

\newcommand{\y}{\mathbf{y}}
\newcommand{\z}{\mathbf{z}}

\newcommand{\mb}[1]{\mathbb{#1}}

\newtheorem{thm}{Theorem}

\newtheorem{lem}{Lemma}

\newtheorem{defn}{Definition}

\newtheorem{assum}{Assumption}

\ifodd 1

\newcommand{\com}[1]{{\color{red}\textbf{Comment}: #1}}
\newcommand{\comtoo}[1]{{\color{purple}\textbf{Comment}: #1}}
\newcommand{\resp}[1]{{\color{cyan}\textbf{Response}: #1}} 
\else

\newcommand{\com}[1]{}
\newcommand{\comtoo}[1]{}
\newcommand{\resp}[1]{}
\fi

\allowdisplaybreaks

\def\BibTeX{{\rm B\kern-.05em{\sc i\kern-.025em b}\kern-.08em
    T\kern-.1667em\lower.7ex\hbox{E}\kern-.125emX}}
\begin{document}
\pagenumbering{arabic}

\title{Providing Differential Privacy for Federated Learning Over Wireless: A Cross-layer Framework
\thanks{This work is supported in part by NSF CNS-2112471.}
\thanks{This paper was presented in part at IEEE International Conference on Communications (ICC), June 2024~\cite{mao2024privacy}.}
\thanks{J. Mao and A. Yener are with the Department of Electrical and Computer Engineering, The Ohio State University, Columbus, OH 43210 USA (e-mail: mao.518@osu.edu; yener@ece.osu.edu).}
\thanks{T. Yin and M. Liu are with the Dept. of Electrical Engineering and Computer Science, University of Michigan, Ann Arbor, MI 48109 USA (e-mail: tyin@umich.edu; mingyan@umich.edu).}}

\author{Jiayu Mao, {\it Student Member, IEEE}, Tongxin Yin, Aylin Yener, {\it Fellow, IEEE}, and Mingyan Liu, {\it Fellow, IEEE}
}

\newgeometry{left=0.625in, right=0.625in,top=0.75in,bottom=1in}
\maketitle

\begin{abstract}
Federated Learning (FL) is a distributed machine learning framework that inherently allows edge devices to maintain their local training data, thus providing some level of privacy. 
However, FL's model updates still pose a risk of privacy leakage, which must be mitigated. 
Over-the-air FL (OTA-FL) is an adapted FL design for wireless edge networks that leverages the natural superposition property of the wireless medium. 
We propose a wireless physical layer (PHY) design for OTA-FL which improves differential privacy (DP) through a fully decentralized, dynamic power control strategy that utilizes both inherent Gaussian noise in the wireless channel and a cooperative jammer (CJ) for additional artificial noise generation when higher privacy levels are required.
Although primarily implemented within the Upcycled-FL framework, where a resource-efficient method with first-order approximations is used at every even iteration to decrease the required information from clients, our power control strategy is applicable to any FL framework, including FedAvg and FedProx as shown in the paper. 
This adaptation showcases the flexibility and effectiveness of our design across different learning algorithms while maintaining a strong emphasis on privacy.
Our design removes the need for client-side artificial noise injection for DP, utilizing a cooperative jammer to enhance privacy without affecting transmission efficiency for higher privacy demands.
Privacy analysis is provided using the Moments Accountant method.
We perform a convergence analysis for non-convex objectives to tackle heterogeneous data distributions, highlighting the inherent trade-offs between privacy and accuracy. 
Numerical results confirm that our approach with different FL algorithms outperforms the state-of-the-art under the same DP conditions on the non-i.i.d. dataset FEMNIST. 
The results also demonstrate the effectiveness of cooperative jammer in meeting stringent privacy demands.
\end{abstract}

\section{Introduction}
\label{sec:intro}
Federated Learning (FL)~\cite{mcmahan2017communication} offers a decentralized approach to machine learning where multiple clients, each with its own private dataset, collaboratively train a global model without data centralization. 
Facilitated by a parameter server (PS), FL involves iterative updates where clients train models locally and then transmit these models to the PS. 
The PS aggregates these individual model updates to generate a new global model.
Although promising, implementation of FL in mobile edge networks requires careful orchestration of resources.

\subsection{Related Work}
Over-the-air federated learning (OTA-FL)~\cite{amiri2020machine}, tailors federated learning to wireless networks by exploiting the superposition property of wireless channels. 
This approach enables clients to simultaneously transmit their local model updates over the air, allowing the PS to directly receive the corresponding update aggregation. 
Utilizing the superposition property, OTA-FL supports simultaneous client transmissions, significantly reducing communication delays and enhancing the scalability of federated learning applications, for example, removing the need for client selection in some cases.
While OTA-FL offers efficiency and data privacy by processing data locally, it still poses potential privacy risks. 
The transmitted analog signals containing model updates could inadvertently reveal sensitive information about the local data. Additionally, the iterative nature of the learning process can amplify this risk, as repeated data usage may lead to increased cumulative information exposure. 
This necessitates the implementation of enhanced security measures to safeguard against privacy leakage throughout the training process.

Differential privacy (DP)~\cite{dwork2014algorithmic} has emerged as a popular method for ensuring individual privacy and serves as a standard privacy guarantee of distributed learning algorithms.
To obtain DP, a natural approach is to add a degree of uncertainty—typically, random Gaussian noise—into released statistics, such as trained models, thereby obscuring the influence of any single data point.
References~\cite{seif2020wireless,sonee2021wireless,park2023differential} incorporate artificial noise into the FL local model updates before transmission to preserve privacy.
Similarly,~\cite{liao2022over} adds correlated perturbations on transmitted signals.
In~\cite{hasirciouglu2021private}, anonymity of over-the-air computation is utilized to reduce the amount of artificial noise.
In~\cite{mohamed2021privacy}, in addition to additive noise, user sampling is used to enhance privacy guarantee. 
In~\cite{yan2022private}, a misaligned power allocation is proposed to minimize the optimality gap of DP OTA-FL.
In~\cite{jiang2022optimized}, power allocation for local updates and injected noise is jointly optimized.
There, a privacy-accuracy trade-off arises as introducing more noise enhances privacy but potentially degrades the model's accuracy.

OTA-FL can potentially achieve a certain level of DP at no additional cost by leveraging inherent channel noise through appropriate power control design.
Instead of static power allocation,~\cite{liu2020privacy} proposes an adaptive power control approach whereby privacy can be obtained ``for free'' when a relatively large  privacy budget is allowed.
An energy-efficient method is developed to reduce the transmit power~\cite{koda2020differentially}.
Reference~\cite{zhang2023communication} includes a compression technique for updates when utilizing channel noise for intrinsic privacy.
Reference~\cite{yan2023over} studies a device scheduling scheme for differential private OTA-FL.
Reference~\cite{wei2023differentially} proposes an orthogonal transmission method to fully utilize transmit power.

\subsection{Contributions}
In this work, we propose to leverage the wireless medium to achieve any $(\varepsilon, \delta)$-differential privacy requirement.
We introduce a fully decentralized dynamic power control, which relies on real-time local data from each client in each global round. 
To enhance its efficacy, we integrate an adaptive power control strategy that guarantees differential privacy while ensuring robust performance in learning.
Implemented by each client in each iteration, this strategy can work with any federated learning framework.
To focus our efforts, we first build on our previous work Upcycled-FL~\cite{yin2023federated} which updates local models only during odd iterations and uses first-order approximations in even iterations to reduce potential information leakage.
In addition, we apply our proposed power control design to FedAvg~\cite{mcmahan2017communication} and FedProx~\cite{li2020federated}, two of the most widely used FL frameworks, to verify its adaptability across various FL frameworks.
We further strengthen our system's privacy by introducing a cooperative jammer (CJ)~\cite{tekin2007gaussian,tekin2008general}, initially developed for PHY security in multiple access wiretap channels in our previous work, for providing confidentiality. 
We employ the CJ to generate additional artificial noise, improving privacy without reducing clients' transmission efficiency, we would activate the CJ only when the inherent channel noise is insufficient for the privacy requirements.
The CJ can be a dedicated node in the network or an external edge device who simply signs up to aid in privacy preserving FL in exchange access to the global model.
We assume that the clients have perfect channel state information (CSI).
Instead of the standard DP analysis, we employ the Moments Accountant~\cite{abadi2016deep} method for a more accurate and stringent analysis of $(\varepsilon, \delta)$-DP requirement. 
We also develop an artificial noise design for scenarios demanding stringent privacy.
More general than the standard convex functions and i.i.d. data assumptions in the literature, we conduct our convergence analysis for non-convex objectives and heterogeneous data distributions, and highlight the inherent trade-offs between privacy and accuracy.
Our numerical evaluations reveal that our power control strategy, when applied to Upcycled-FL, FedAvg, and FedProx, outperforms the current state-of-the-art~\cite{liu2020privacy} with the same DP guarantee, even in the absence of a CJ.
Additionally, our results demonstrate the efficacy of our proposed artificial noise design in meeting high privacy requirements on the non-i.i.d. dataset FEMNIST.

\subsection{Organization}
The remainder of the paper is organized as follows. 
We describe the system model under consideration in Section~\ref{sec: sysmodel}.
We introduce the decentralized adaptive power control design in Section~\ref{subsec:pcdesign}, followed by its privacy analysis and cooperative jammer noise design in Section~\ref{subsec:privacytheorem}.
We provide the convergence analysis in Section~\ref{sec: conv}. We present numerical results in Section~\ref{sec: exp}. Section~\ref{sec: conclusion} concludes the paper.
Proofs are provided in Appendices A-C.

\section{System Model} 
\label{sec: sysmodel}
\subsection{Federated Learning Model} 
\label{subsec: fl}
We consider a federated learning system that consists of a parameter server (PS) and various clients in a set $\Ic$.
Each client has a unique local dataset $\Dc_i$.
Datasets are non i.i.d.
The local loss function of the $i$-th client is defined as
\begin{equation}
    F_i(\w; \Dc_i) \triangleq \frac{1}{| \Dc_i |} \sum_{\xi^i_j \in \Dc_i} F(\w, \xi^i_j)
\end{equation}
where $\w \in \mathbb{R}^d$ is the model parameter, $\xi^i_j$ the $j$-th data sample from local dataset $\Dc_i$, and $F(\w, \xi^i_j)$ is the sample-wise loss function.
The goal of FL is to minimize the global loss function:
\begin{equation}
    \min_{\w \in \mathbb{R}^d}f(\w) \triangleq \min_{\w\in\mathbb{R}^d} \sum_{i \in \Ic} p_i F_i(\w; \Dc_i) = \Eb [F_i(\w; \Dc_i)], 
    \label{eq: objective}
\end{equation}
where $p_i = \frac{|\Dc_i|}{\sum_{i \in \Ic} |\Dc_i|}$ represents the proportion of data samples originating from client $i$, $\Eb[\cdot]$ is the expectation taken across all clients.
In this work, we will assume general non-convex objectives, i.e., local loss functions $F_i(\w, \Dc_i)$ are non-convex, reflecting real-world scenarios.
Additionally, the volume of training data differs across clients, meaning $p_i \ne p_j$ for $i\ne j$.

In FL, the training of the model is completed through a series of iterative steps.
During each iteration, clients execute local computations to update their model parameters based on individual datasets.
Following this, the local updates are sent to the PS.
Upon receiving these updates, the PS aggregates and combines them into a new global model by a weighted average.
Notably, in OTA-FL, the superposition characteristic of the wireless channel is exploited, allowing simultaneous transmission and aggregation of updates as clients transmit concurrently.
After computing the global update, the PS broadcasts it to all clients, model starting the next iteration of training.
This iterative process continues until convergence of the global model is achieved.

In this paper, we consider three FL algorithms: two established baseline frameworks, FedAvg~\cite{mcmahan2017communication} and FedProx~\cite{li2020federated}, and the third is Upcycled-FL, which is an FL framework that minimizes information leakage and computational requirements, as elaborated in our prior work~\cite{yin2023federated}. 
We will describe the implementation of these algorithms, focusing specifically on their local and global update mechanisms.
We consider $M$ total iterations for FedAvg and FedProx, and $2M$ total iterations for Upcycled-FL, for a fair comparison.

\subsubsection{Local update}

For FedAvg, local objectives are just local loss functions~\cite{mcmahan2017communication}; while for FedProx and Upcycled-FL, the local objectives are loss functions with proximal terms~\cite{li2020federated}:
\begin{equation}
\min_{\w \in \mathbb{R}^d} g_i(\w; \Bar{\w}) \triangleq F_i(\w; \Dc_i) + \frac{\mu}{2} \|\w - \Bar{\w}\|^2,
\label{eq: objective_local}
\end{equation}
where $\Bar{\w}$ denotes the global model. FedAvg and FedProx perform local updates at every global iteration, whereas Upcycled-FL executes local updates only during odd iterations.

\begin{itemize}
    \item FedAvg: at each global iteration $m$,
    \begin{equation}
        \w_i^m = \arg\min F_i(\w; \Dc_i).
    \end{equation}
    \item FedProx: at each global iteration $m$,
    \begin{equation}
        \w_i^m = \arg\min g_i(\w; \Bar{\w}^m).
    \end{equation}
    \item Upcycled-FL: at each odd iteration $2m-1$,
    \begin{equation}
    \w_i^{2m-1} = \arg\min g_i(\w; \Bar{\w}^{2m-2}). \label{equ:local update}
    \end{equation}
\end{itemize}

In this work, we employ stochastic gradient descent (SGD) for local updates. While alternative optimizers could also be used, they are not the focus of this work.

\subsubsection{Global update}
To further reduce communication overhead, clients transmit their updates or gradients rather than the complete local models.
On the server side, both FedAvg and FedProx conduct a single global step. By contrast, Upcycled-FL implements two global steps by reusing the intermediate updates and applying a first-order approximation to the even iterations.
\begin{itemize}
    \item FedAvg and FedProx: at global iteration $m$,
    \begin{equation} 
    \Bar{\w}^{m} = \Bar{\w}^{m-1} + \sum_{i \in \Ic} p_i (\w_i^{m} - \Bar{\w}^{m-1}).
    \end{equation}
    \item Upcycled-FL: at global iteration $2m-1$ and $2m$,
    \begin{equation} \label{equ: global update odd}
    \Bar{\w}^{2m-1} = \Bar{\w}^{2m-2} + \sum_{i \in \Ic} p_i (\w_i^{2m-1} - \Bar{\w}^{2m-2}),
    \end{equation}
    \begin{equation} \label{equ: global update even}
    \Bar{\w}^{2m} = \Bar{\w}^{2m-1} + \frac{\mu}{\mu + \lambda_m} (\Bar{\w}^{2m-1} - \Bar{\w}^{2m-2}).
    \end{equation}
\end{itemize}

The transmission of the updates from clients to the server could potentially expose information about the local datasets.
This motivates the adoption of DP that offers privacy assurances independent of the server's computation and data processing resources.
It should be noted that in the case of Upcycled-FL, information leakage occurs only during odd iterations because $\mathcal{D}_i$ is not utilized in even iterations. This contributes to enhanced privacy protection and improves the trade-off between privacy and accuracy.
The details will be provided in Sec.~\ref{subsec:privacytheorem}.

\subsection{Communication Model} 
\label{subsec: comm}
\begin{figure}[t]
    \centering
    \includegraphics[scale=0.25]{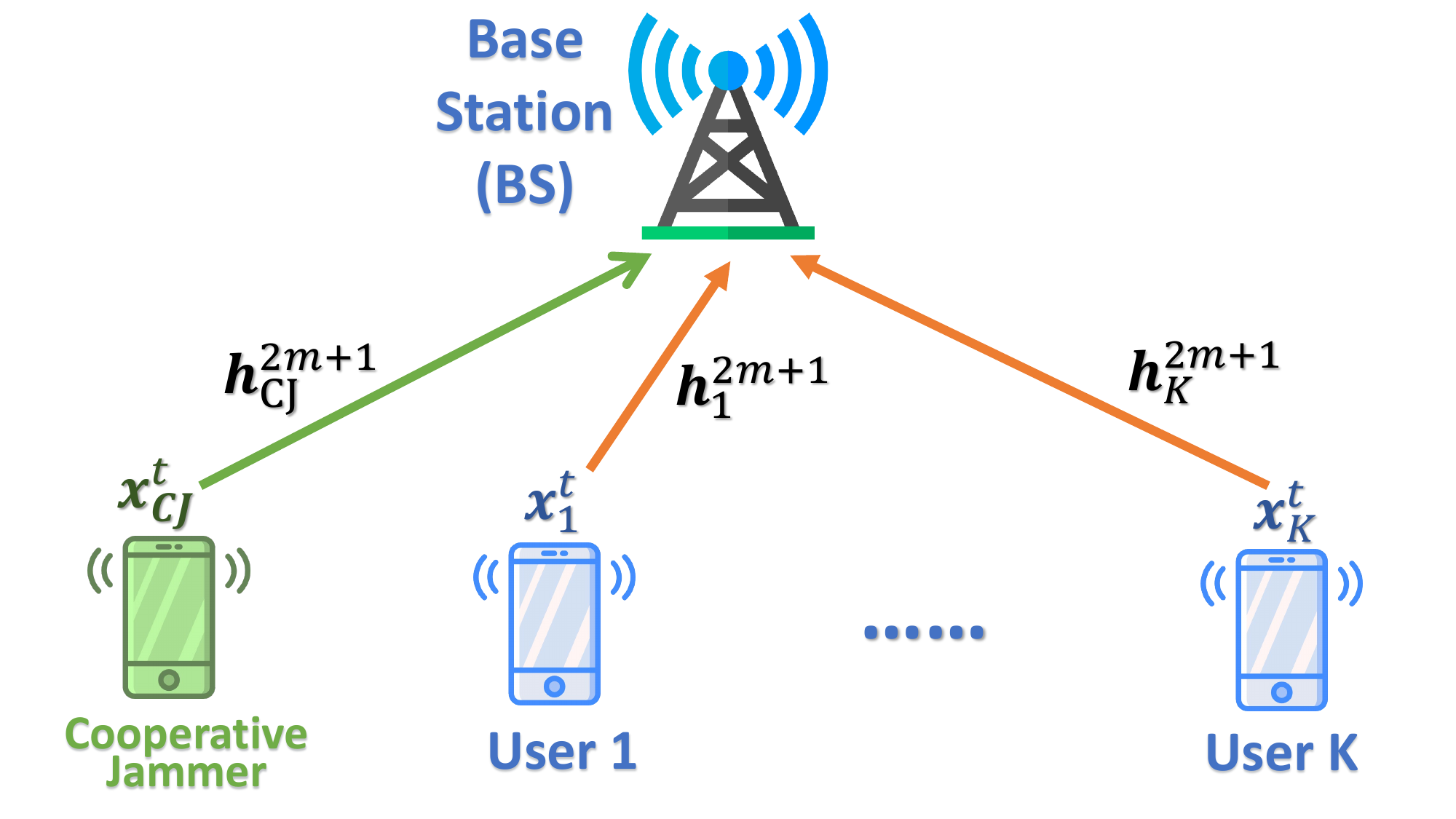}
    \caption{The differential private over-the-air communication system.}
    \label{fig:systemmodel}
\end{figure}

We consider an over-the-air federated learning (OTA-FL) system as illustrated in Fig.~\ref{fig:systemmodel}.
The setup includes a base station serving as the parameter server, $K$ users, and a cooperative jammer (CJ), each equipped with a single antenna.
The CJ's role is specifically to transmit artificial noise when needed, i.e., the privacy requirement is too stringent for the wireless medium itself to provide.
We consider that the downlink communication is synchronous and error-free, i.e., clients can receive global model parameters perfectly.
The uplink communication assumes block flat-fading channels, where channel gains are stable within a communication round but vary independently from one iteration to the next. 
Our system employs analog non-orthogonal multiple access (NOMA), which enables over-the-air computation.
Consistent with previous studies on differentially private OTA-FL~\cite{seif2020wireless,sonee2021wireless,park2023differential, liao2022over, hasirciouglu2021private, mohamed2021privacy, yan2022private, jiang2022optimized, koda2020differentially, liu2020privacy, yan2023over, zhang2023communication, wei2023differentially}, we consider that all users have their perfect channel state information (CSI). This allows each client to compensate for its own channel.

Define $\x_i^t \in \mb{R}^d$ as the signal transmitted by client $i$ in iteration $t$ and $\x_{CJ}^t \in \mb{R}^d$ as the artificial noise generated by the cooperative jammer. Consequently, the signal received at the PS is as follows:
\begin{equation}
    \y^{t} = \sum_{i \in \mathcal{I}} h^{t}_i \x_i^{t} + h^{t}_{CJ} \x_{CJ}^{t} + \z^{t},
\end{equation}
where $h^t_i, h^t_{CJ}\in \mb{C}$ denote the channel gains from client $i$ and the cooperative jammer, respectively, $\z^t$ represents i.i.d. additive white Gaussian noise (AWGN) with a mean of zero and a variance of $\sigma_c^2$.

Furthermore, the transmit power constraint for client $i$ is specified as follows:
\begin{equation} \label{inequ:powerconstr}
    \Eb[\| \x_i^t \|^2] \leq P_i, \forall i \in \Ic, \forall t, 
\end{equation}
where $P_i$ is the maximum transmit power of client $i$. 

\subsection{Differential Privacy}
We consider that the parameter server is ``honest-but-curious''. That is, while the server does not actively seek to attack the learning process and is in compliance with all its functions to facilitate the training data, it might attempt to infer information about the local datasets by analyzing the sequence of signals $\{\y^m\}_{m=1}^M$ received across $M$ consecutive rounds.
Differential privacy (DP) provides a mathematical framework for quantifying privacy by limiting the discrepancy between two conditional probability distributions:
$P(\y|\Dc)$ and $P(\y|\Dc')$, where $\y = \{\y^m\}_{m=1}^M$ represents the received signals, and $\Dc, \Dc'$ are two ``neighboring" datasets.
These datasets are considered neighboring if they differ by the alteration of a single data sample in one client within the system.
We denote $\|\Dc - \Dc'\|_1$ as the cardinality difference between the two sets, leading to the following formal definition of differential privacy.

\begin{defn}
    (Differential Privacy~\cite{dwork2014algorithmic}) Consider two neighboring datasets $\Dc$ and $\Dc'$ such that $\|\Dc - \Dc'\|_1=1$. A randomized algorithm $\mathcal{M}: X^n \rightarrow \mathcal{R}$ is $(\varepsilon,\delta)$-differential privacy if for all outputs $\mathcal{S} \subseteq \mathcal{R}$ and all pairs of $\Dc, \Dc'$ we have:
    \begin{equation}
        Pr[\mathcal{M}(\Dc) \in \mathcal{S}] \leq e^{\varepsilon} Pr[\mathcal{M}(\Dc') \in \mathcal{S}] + \delta.
    \end{equation}
\end{defn}

\section{Adaptive Private Power Control} 
\label{sec: alg}

In this section, we present an adaptive power control strategy designed to offer differential privacy guarantees while maintaining robust learning performance. 
Our proposed approach is a fully decentralized method, implementable by each client in each iteration.
Importantly, our proposed power control mechanism is independent of the underlying FL algorithm, meaning it can be seamlessly integrated into any FL framework. 
While our primary focus is on Upcycled-FL, we also apply it to FedAvg and FedProx to demonstrate its versatility across different FL algorithms.

\begin{algorithm}[t]
    \caption{Private Upcycled-OTA-FL} \label{alg:upcycled-ota} 
    \begin{algorithmic}[1]
    \STATE 
    \emph{\bf Initialization: global model $\w_0$, $\lambda_m>0$, $\alpha^m_i$, $\alpha^m_u$, $\alpha^m_{CJ}$,$\tau$, $\forall i \in \Ic$.}
    \FOR{$m=1, \dots, M$}
        \STATE {The server sends $\Bar{\w}^{2m-2}$ to the clients, the clients do the local training by~\eqref{equ:local update}}.
        \STATE {Device $i$ transmits local update to the server by~\eqref{equ:transignal}, and CJ transmits~\eqref{equ:ANG} only when it is needed}.
        \STATE {The server receives the aggregated signal and does two steps of global model update $\Bar{\w}^{2m-1}, \Bar{\w}^{2m}$ by~\eqref{equ: global update odd},~\eqref{equ: global update even}, then broadcasts $\Bar{\w}^{2m}$ to clients}.
    \ENDFOR
    \end{algorithmic}
\end{algorithm}

\subsection{Adaptive Power Control}
\label{subsec:pcdesign}
We explore a dynamic power control (PC) approach implemented by both the server and the clients. 
Initially, we discuss this strategy within the context of the Upcycled-FL framework as shown in Algorithm~\ref{alg:upcycled-ota}, and subsequently, we extend our examination to include FedAvg and FedProx. 
It is important to reiterate that in Upcycled-FL, client participation in global training occurs only during odd iterations.
Following the completion of local training, client $i$ constructs its transmit signal $\x_i^{2m+1}$ in round $2m+1$ as follows:
\begin{equation} \label{equ:transignal}
    \x_i^{2m+1} = \alpha_i^{2m+1} (\w_i^{2m+1} - \Bar{\w}^{2m}),
\end{equation}
where $\alpha^{2m+1}_i$ represents the PC parameter for client $i$, $(\w_i^{2m+1} - \Bar{\w}^{2m})$ denotes the local update.

For simplicity, we assume that the gradient of the local function is bounded as in~\cite{liu2020privacy,seif2020wireless}. 
This condition can be achieved in practice by clipping the updates.
\begin{assum} \label{a_bounds_local} (Bounded Gradient)
    The gradient of local objective is bounded, i.e., $\|\nabla g_i(\w)\|^2 \leq \tau$.
\end{assum}

Note that unlike the majority of existing literature, this work does not involve adding artificial noise to each client's transmissions. 
This enables clients to use their transmit powers solely for sending training information. 
Additionally, with clients having the channel state information (CSI), we can utilize channel inversion.
In our previous work~\cite{mao2022charles}, we have proposed a joint learning and communication design to develop an adaptive OTA-FL framework. Inspired by it,
we design the PC factor for client $i$ as:
\begin{equation} \label{equ: PC_client}
    \alpha^{2m+1}_i = \frac{\alpha^{2m+1}_u p_i}{h^{2m+1}_i \tau s_i^{2m+1}},
\end{equation}
where $\alpha^{2m+1}_u$ represents the PC factor for server, $p_i$ denotes the weight of client $i$, $\tau$ refers to the bound of local gradients defined above, $s_i^{2m+1}$ is employed to meet the transmit power constraints in~\eqref{inequ:powerconstr}.
During each odd iteration $2m+1$, we incorporate~\eqref{equ: PC_client} into~\eqref{equ:transignal} to calculate the appropriate $s_i^{2m+1}$, ensuring compliance with the transmit power constraint as specified in~\eqref{inequ:powerconstr}. 
Naturally, $s_i^{2m+1}$ plays a crucial role in mitigating the impact of fading. 
Moreover, we maintain $s_i^{2m+1} \geq 1$ to bound the privacy level $\varepsilon$, as will be further discussed in the following section. 
It is worth mentioning that our design is unique in that it relies exclusively on data from the current iteration of client $i$, thus supporting a fully decentralized framework that avoids dependence on information from other iterations.

We denote the signal that the cooperative jammer transmits in round $2m+1$ as $\x_{CJ}^{2m+1}$, which generates artificial noise solely. In particular,
\begin{equation} \label{equ:ANG}
    \x_{CJ}^{2m+1} = \alpha_{CJ}^{2m+1} \n_{CJ}^{2m+1}, 
\end{equation}
where $\n_{CJ}^{2m+1}$ represents standard Gaussian noise and $\alpha_{CJ}^{2m+1}$ denotes power control parameter.
We employ $\alpha_{CJ}^{2m+1}$ to adjust the amount of artificial noise introduced into the system. 
The cooperative jammer, positioned among the clients, injects this noise unrecognizable to the ``honest-but-curious" parameter server.
While channel noise naturally adds some level of privacy, as discussed in~\cite{liu2020privacy}, it may prove insufficient under stringent privacy conditions. 
In such cases, we can augment the system with a CJ that delivers additional artificial noise to fulfill the privacy demand.
Practically, the CJ can be an external entity or its role can rotate among the clients.
Specifically, in each iteration, a designated client can skip local training to send artificial noise to the server, with a different client taking up this role in each iteration. The design details of $\alpha_{CJ}^{2m+1}$ will be elaborated in the following section.

Finally, PS adjusts the received signal by applying its power control factor $\alpha_u^{2m+1}$, resulting in the following updated global model:
\begin{align}
    \Bar{\w}^{2m+1} = \Bar{\w}^{2m} &+  \sum_{i \in \mathcal{I}} \frac{h^{2m+1}_i}{\alpha^{2m+1}_u} \x_i^{2m+1} \nonumber\\
    &+ \frac{h^{2m+1}_{CJ} \alpha_{CJ}^{2m+1}}{\alpha^{2m+1}_u}  \n_{CJ}^{2m+1} + \frac{\z^{2m+1}}{\alpha^{2m+1}_u}.
\end{align}

As can be seen from the development, the proposed power control strategy is independent of any particular learning algorithm. 
We thus also consider FedAvg and FedProx as other examples. 
The primary modification involves adjusting the iteration labeling from $2m+1$ in odd iterations to $m$ for every global iteration, reflecting the fact that clients transmit local updates in each iteration in FedAvd and FedProx. Specifically, during global iteration $m$,
\begin{equation} \label{equ: PC_client_fedavg}
    \alpha^{m}_i = \frac{\alpha^{m}_u p_i}{h^{m}_i \tau s_i^{m}},
\end{equation}
\begin{equation} \label{equ:ANG_fedavg}
    \x_{CJ}^{m} = \alpha_{CJ}^{m} \n_{CJ}^{m}, 
\end{equation}
\begin{align}
    \Bar{\w}^{m+1} = \Bar{\w}^{m} &+  \sum_{i \in \mathcal{I}} \frac{h^{m}_i}{\alpha^{m}_u} \x_i^{m} + \frac{h^{m}_{CJ} }{\alpha^{m}_u}  \x_{CJ}^{m} + \frac{\z^{m}}{\alpha^{m}_u}.
\end{align}

\subsection{Privacy Analysis}
\label{subsec:privacytheorem}
We provide the privacy analysis and quantify the overall privacy loss associated with our proposed approach. 
Again, we focus on Upcycled-FL, and subsequently extend the results to FedAvg and FedProx.
It should be noted that within the Upcycled-FL framework, privacy leakage happens exclusively during odd iterations, which are the instances when local datasets are in use.
To provide a comprehensive result, we recapitulate the following lemma from our prior work~\cite{yin2023federated}:
\begin{lem} \label{lemma:1}
    For any integer $m \geq 1$, if the total privacy loss is bounded by $\varepsilon_m$ up to the $(2m-1)$-th iteration, it remains bounded through the $2m$-th iteration.
\end{lem}

We employ the moments accountant approach~\cite{abadi2016deep}, since it allows for a more accurate assessment of the cumulative privacy loss, especially for intricate, iterative processes that traditional composition theorems cannot adequately address.
Leveraging Lemma~\ref{lemma:1}, we specify the following privacy level:
\begin{thm} \label{thm:privacy}
    Consider the adaptive power control with Algorithm~\ref{alg:upcycled-ota} over $2M$ iterations, then for any $\delta \in [0,1]$, the algorithm is $(\varepsilon,\delta)$-differential privacy for client $i$ for 
    \begin{align}
        &\varepsilon = 2 \sqrt{ \frac{1}{2|\Dc|^2} \sum_{m=1}^M \frac{1}{(s_i^{2m+1})^2(\sigma^{2m+1})^2} \log \left(\frac{1}{\delta}\right)} \nonumber\\
        &+ \frac{1}{2|\Dc|^2} \sum_{m=1}^M \frac{1}{(s_i^{2m+1})^2(\sigma^{2m+1})^2}  \\
        & \leq 2 \sqrt{ \frac{1}{2|\Dc|^2} \sum_{m=1}^M \frac{1}{(\sigma^{2m+1})^2} \log \left(\frac{1}{\delta}\right)} + \frac{1}{2|\Dc|^2} \sum_{m=1}^M \frac{1}{(\sigma^{2m+1})^2} \label{inequ:privacy_bound}
    \end{align}
    where $|\Dc| = \sum_{i \in \Ic} |\Dc_i|$ is the total number of training data points, $(\sigma^{2m+1})^2$ is variance of the equivalent additive noise, $s_i^{2m+1}$ is parameter to satisfy transmit power constraint.
\end{thm}
\begin{proof}
The proof is detailed in Appendix A.
\end{proof}

Theorem~\ref{thm:privacy} shows that the privacy level is influenced by the total size of the training dataset, the cumulative effective noise per iteration, and a constant ensuring compliance with communication constraints. The effective noise, $(\sigma^{2m+1})^2$, is composed of two elements: channel noise and artificial noise. As previously noted, channel noise alone offers some level of privacy protection.
We now proceed to detail the design of the artificial noise, i.e., the cooperative jamming signal.
\begin{thm} \label{thm:PC_{CJ}Ndevice}
    Consider the adaptive power control with Algorithm~\ref{alg:upcycled-ota} over $2M$ iterations, then the algorithm is $(\varepsilon,\delta)$-differential private for any $\varepsilon > 0$ if the power control factor for cooperative jammer is designed as
    \begin{equation}
        \alpha^{2m+1}_{CJ} \geq \frac{\alpha^{2m+1}_u}{h^{2m+1}_{CJ}} \sqrt{ \frac{M}{2|\Dc|^2 a^2} \log \frac{1}{\delta} - \frac{\sigma_{c}^2}{(\alpha^{2m+1}_u)^2}},
    \end{equation}
    \begin{equation}
        where \quad a = -\log \frac{1}{\delta} + \sqrt{\left(\log \frac{1}{\delta}\right)^2 + \varepsilon \log \frac{1}{\delta}},
    \end{equation}
    where $|\Dc| = \sum_{i \in \Ic} |\Dc_i|$ is the total number of data points.
\end{thm}
\begin{proof}
    See Appendix B.
\end{proof}

Theorem~\ref{thm:PC_{CJ}Ndevice} highlights the role of channel noise, specifying that if $\frac{M}{2|\Dc|^2 a^2} \log \frac{1}{\delta} < \frac{\sigma_{c}^2}{(\alpha^{2m+1}_u)^2}$, then the transmission of artificial noise is unnecessary as channel noise alone meets the required privacy standards.

We now apply the aforementioned results to FedAvg and FedProx. The primary modification involves the iteration labeling, shifting from $2m+1$ for odd iterations to $m$ for every global iteration, to accommodate the clients now transmitting local updates in each iteration. For clarity and completeness, we reiterate this adjustment here.

\begin{thm} \label{thm:privacy_fedavg}
    Consider the adaptive power control with FedAvg or FedProx over $M$ iterations, then for any $\delta \in [0,1]$, the algorithm is $(\varepsilon,\delta)$-differential privacy for client $i$ for 
    \begin{align}
        &\varepsilon = 2 \sqrt{ \frac{1}{2|\Dc|^2} \sum_{m=1}^M \frac{1}{(s_i^{m})^2(\sigma^{m})^2} \log \left(\frac{1}{\delta}\right)} \nonumber\\
        &+ \frac{1}{2|\Dc|^2} \sum_{m=1}^M \frac{1}{(s_i^{m})^2(\sigma^{m})^2}  \\
        & \leq 2 \sqrt{ \frac{1}{2|\Dc|^2} \sum_{m=1}^M \frac{1}{(\sigma^{m})^2} \log \left(\frac{1}{\delta}\right)} + \frac{1}{2|\Dc|^2} \sum_{m=1}^M \frac{1}{(\sigma^{m})^2} \label{inequ:privacy_bound}
    \end{align}
    where $|\Dc| = \sum_{i \in \Ic} |\Dc_i|$ is the total number of training data points, $(\sigma^{m})^2$ is variance of the equivalent additive noise, $s_i^{m}$ is parameter to satisfy transmit power constraint.
\end{thm}
\begin{proof}
Change the iteration label from $2m$ to $m$, then the proof is same as that in Appendix A.
\end{proof}
\begin{thm} \label{thm:PC_{CJ}Ndevice_fedavg}
    Consider the adaptive power control with FedAvg or FedProx over $M$ iterations, then the algorithm is $(\varepsilon,\delta)$-differential privacy for any $\varepsilon > 0$ if the power control factor for cooperative jammer is designed as
    \begin{equation}
        \alpha^{m}_{CJ} \geq \frac{\alpha^{m}_u}{h^{m}_{CJ}} \sqrt{ \frac{M}{2|\Dc|^2 a^2} \log \frac{1}{\delta} - \frac{\sigma_{c}^2}{(\alpha^{m}_u)^2}},
    \end{equation}
    \begin{equation}
        where \quad a = -\log \frac{1}{\delta} + \sqrt{\left(\log \frac{1}{\delta}\right)^2 + \varepsilon \log \frac{1}{\delta}},
    \end{equation}
    where $|\Dc| = \sum_{i \in \Ic} |\Dc_i|$ is the total number of data points.
\end{thm}
\begin{proof}
    Same steps as in Appendix B.
\end{proof}

\section{Convergence Analysis}
\label{sec: conv}
We begin our discussion with a theoretical convergence analysis for Algorithm~\ref{alg:upcycled-ota}, then extend the analysis to FedAvg and FedProx.
To begin, we adopt the standard assumptions used in our previous work~\cite{yin2023federated}.

\begin{defn} ($B$-Dissimilarity) \label{d_dissimilar}
    The local loss function $F_i$ is $B$-dissimilar if $\ \Eb [\|\nabla F_i(\w) \|^2] \leq \|\nabla f(\w) \|^2 B^2, \forall \w$, where the expectation is over clients.
\end{defn}
\begin{lem} \label{lemma_boundg}
    There exists $B$ such that $F_i$ is $B$-dissimilar if $\|\nabla F_i(\w) - \nabla f(\w) \| \leq \kappa_i, \forall \w, \forall i,$ for some $\kappa_i$.
\end{lem}
The parameter $B$ quantifies the statistical heterogeneity among clients, specifically referring to the non-i.i.d. data distribution.
\begin{assum}($L$-Lipschitz Smooth) \label{a_smooth}
	$\exists L > 0$, such that $ \| \nabla F_i(\w_1) - \nabla F_i(\w_2) \| \leq L \| \w_1 - \w_2 \|$, $\forall \w_1, \w_2$, $\forall i \in \Ic$.
\end{assum}
\begin{assum} \label{a_strongconv} (Strongly Convex Local Function)
    $\forall i,$ $g_i(\w;\Bar{\w}^t) = F_i(\w; \Dc_i) + \frac{\mu}{2} \| \w - \Bar{\w}^t\|^2$ is $\rho$-strongly convex.
\end{assum}
\begin{assum} \label{a_bounds} (Bounded Norms)
    $\|\Bar{\w}^{2m-1} - \Bar{\w}^{2m-2} \| \leq q, \forall m$, and $\|\nabla f(\w)\| \leq G , \forall \w$.
\end{assum}

We should note that the assumption of strong convexity is not made directly on the local loss function $F_i(\w;\Dc_i)$ but rather on the augmented function $F_i(\w; \Dc_i) + \frac{\mu}{2} \| \w - \Bar{\w}^t\|^2$. By selecting a sufficiently large $\mu$, this condition can be satisfied. Importantly, as demonstrated in Section~\ref{sec: exp}, our algorithm achieves convergence even without this assumption.

\begin{thm} \label{thm:convergence}
    (Convergence rate) Under Assumptions~\ref{a_smooth}-\ref{a_bounds}, if $\mathbf{C_1}>0$, over $2M$ global rounds we have
    \begin{multline}
        \min_{m \in [M]} \mb{E} \| \nabla f(\Bar{\w}^{2m-1}) \|^2 \leq \frac{f(\Bar{\w}^0) - f(\Bar{\w}^*)}{M \mathbf{C_1}} + \frac{\mathbf{C}_6}{ \mathbf{C_1}} \\
        + \frac{1}{M \mathbf{C_1}} \sum_{m=1}^M \mathbf{C}_2^m + \mathbf{C}_3^m + \mathbf{C}_4^m + \mathbf{C}_5^m,
    \end{multline}
    where $\alpha_u^{2m+1} = \alpha_u$, $\Eb_i[\cdot]$ is expectation over clients,
    \begin{align}
        &\mathbf{C}_1 = \frac{1}{2\mu} - \frac{LB}{\mu \rho^2}, \\
        &\mathbf{C}_2^m = \frac{\mu}{\mu+\lambda_m} \left(1+\frac{2LB(L+\rho)}{\mu\rho^2} \right) qG, \\
        &\mathbf{C}_3^m = \left( \frac{\mu}{\mu+\lambda_m}\right)^2 L \left(1+\frac{(L+\rho)^2}{\mu \rho^2}\right) q^2, \\
        &\mathbf{C}_4^m = \Eb_i \left[\left( \frac{2\mu(\tau s_i^{2m+1}-1)^2 + (2L-\mu)\mu }{2 \mu^2 (\tau s_i^{2m+1})^2}\right) (\kappa_i + G)^2\right], \\
        &\mathbf{C}_5^m = \frac{Ld}{2} \left( \frac{h_{CJ}^{2m+1} \alpha_{CJ}^{2m+1}}{\alpha_u}\right)^2, \\
        &\mathbf{C}_6 = \frac{L\sigma_c^2d}{2\alpha_u^2}.
    \end{align}
\end{thm}
\begin{proof}
See Appendix C.
\end{proof}

Theorem~\ref{thm:convergence} identifies several factors influencing convergence. 
It illustrates that the tunable parameters $\mu, \lambda_m$ of Algorithm~\ref{alg:upcycled-ota}, along with the system's heterogeneity, are encompassed by terms $\mathbf{C}_2^m$ and $\mathbf{C}_3^m$.
Our power control design, as captured in $\mathbf{C}_4^m$, is closely coupled with the dissimilarity among clients.
Meanwhile, $\mathbf{C}_6$ represents the error due to channel noise, and $\mathbf{C}_5^m$ reflects the error from artificial noise, illustrating the trade-off between privacy and accuracy: higher privacy levels require more noise, which in turns challenges convergence and learning performance.

\textbf{Convergence of FedAvg and FedProx.} 
As previously noted, our power control strategy works with any FL algorithm. Consequently, the convergence of our proposed power control within the FedAvg or FedProx frameworks is guaranteed. Convergence can be demonstrated using similar methods as in~\cite{yang2022over,mao2022charles}. We omit the detailed proof here, as the steps are similar and well-established.

\section{Numerical Results} 
\label{sec: exp}
\begin{figure}[t]
    \centering
    \includegraphics[width=0.44\textwidth]{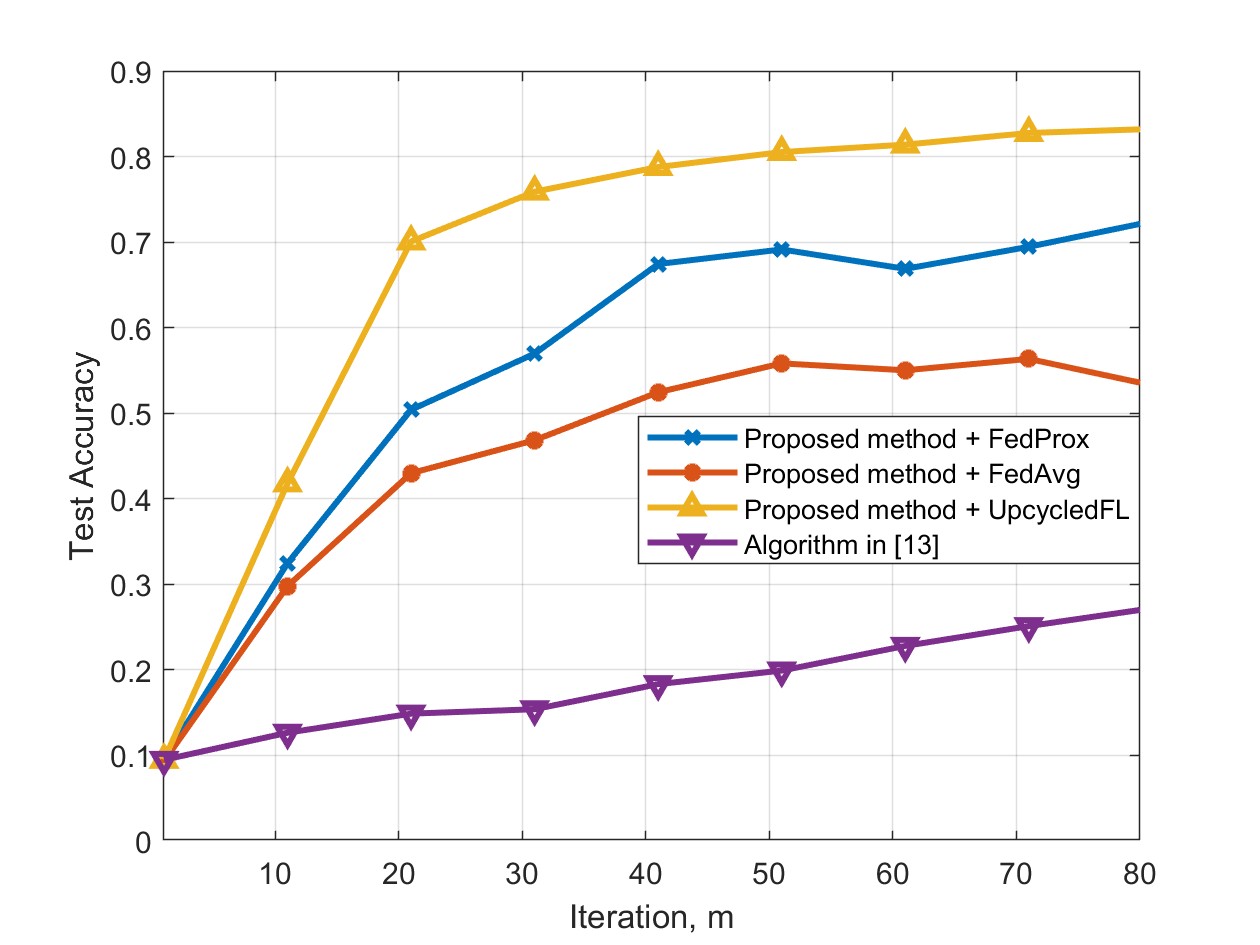}
    \caption{Test accuracy without cooperative jammer $(\varepsilon = 6.52, \delta = 10^{-5})$.}
    \label{fig:testacc_1}
\end{figure}
\begin{figure}[t]
    \centering
    \includegraphics[width=0.44\textwidth]{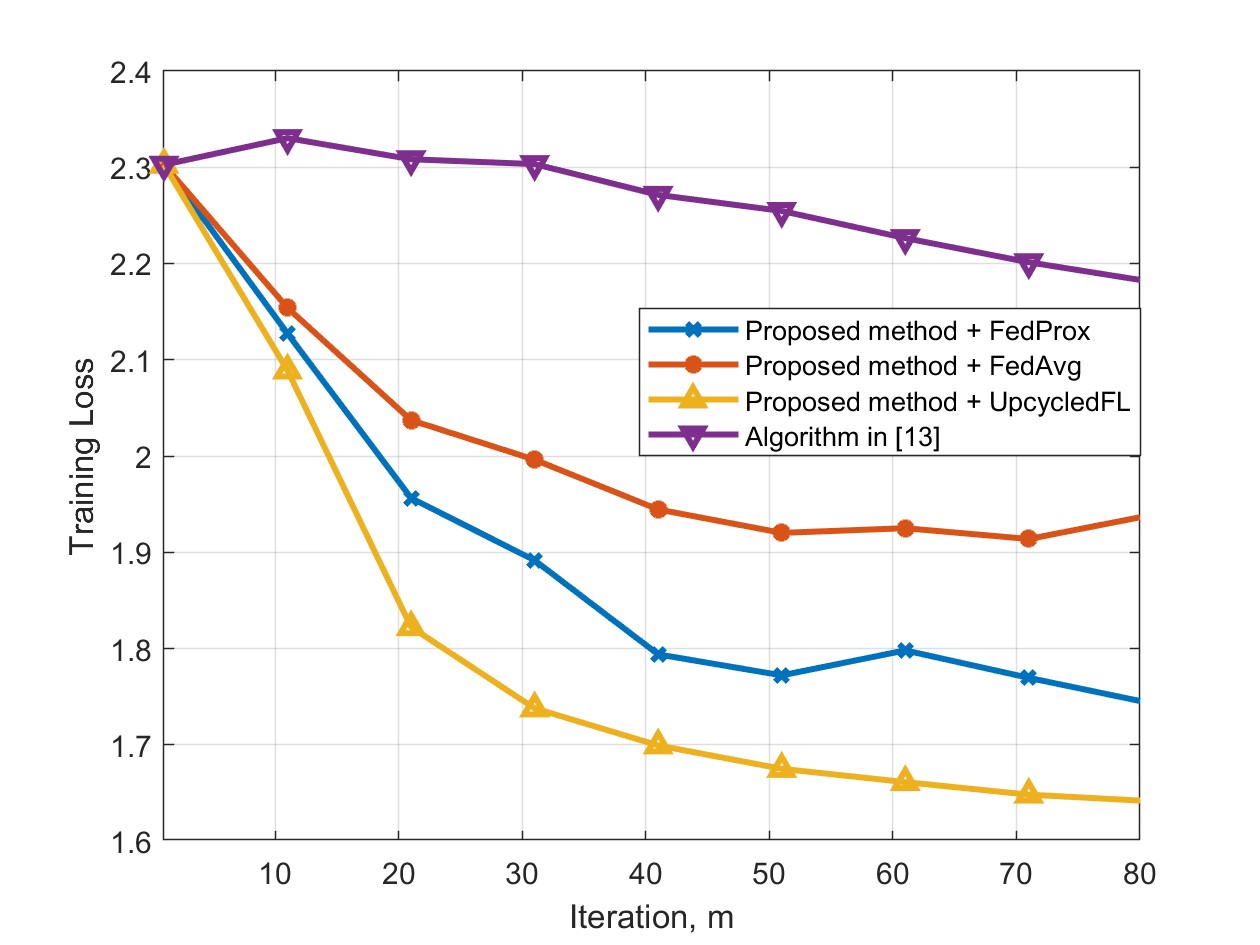}
    \caption{Training loss without cooperative jammer $(\varepsilon = 6.52, \delta = 10^{-5})$.}
    \label{fig:trainingloss_1}
\end{figure}
In this section, we conduct experiments to evaluate the effectiveness of our proposed scheme.
First, we benchmark the learning performance against baseline models that do not incorporate a cooperative jammer (CJ), aligning with the system described in~\cite{liu2020privacy}. Subsequently, we introduce a CJ and evaluate our proposed power control strategy across various federated learning frameworks, under different levels of differential privacy requirements.

Our experiments primarily focus on non-i.i.d. data distributions. 
However, it is well known that FedAvg does not perform well with non-i.i.d data. Therefore, in Section~\ref{subsec:exp-nonIID}, we also conduct experiments on an i.i.d. synthetic dataset. This allows us to fairly compare our proposed power control scheme with the adaptive power control method applied to FedAvg in~\cite{liu2020privacy}.

We evaluate the performance of the following algorithms by comparing it against several established baselines:
\begin{enumerate}
    \item Our proposed power control applied with Upcycled-FL (Algorithm~\ref{alg:upcycled-ota}).
\item  Our proposed power control strategy with FedProx~\cite{li2020federated} framework.
\item Our proposed power control method applied to the FedAvg~\cite{mcmahan2017communication} algorithm.
\item Baseline: Algorithm from~\cite{liu2020privacy}, an adaptive power control approach in a NOMA system using the FedAvg framework.
\end{enumerate}

In our numerical results, we model a wireless FL system comprising $K=50$ users and a cooperative jammer.
We assume block Raleigh fading channels, characterized by $h_{CJ}^{m}, h_i^{m} \sim \mathcal{CN}(0,1)$ for all transmission iterations $m$ and for all users $i$ in the set $\mathcal{I}$.
For consistency and fair comparison with prior studies, specifically~\cite{liu2020privacy}, we adopt the same signal-to-noise ratio (SNR) definition: $SNR_i = P_i/d \sigma_c^2$, where $d$ represents the model dimension. We set maximum $SNR_i$ as 1dB.

\subsection{Non-IID Data}
\label{subsec:exp-nonIID}

\begin{figure}[t]
    \centering
    \includegraphics[width=0.44\textwidth]{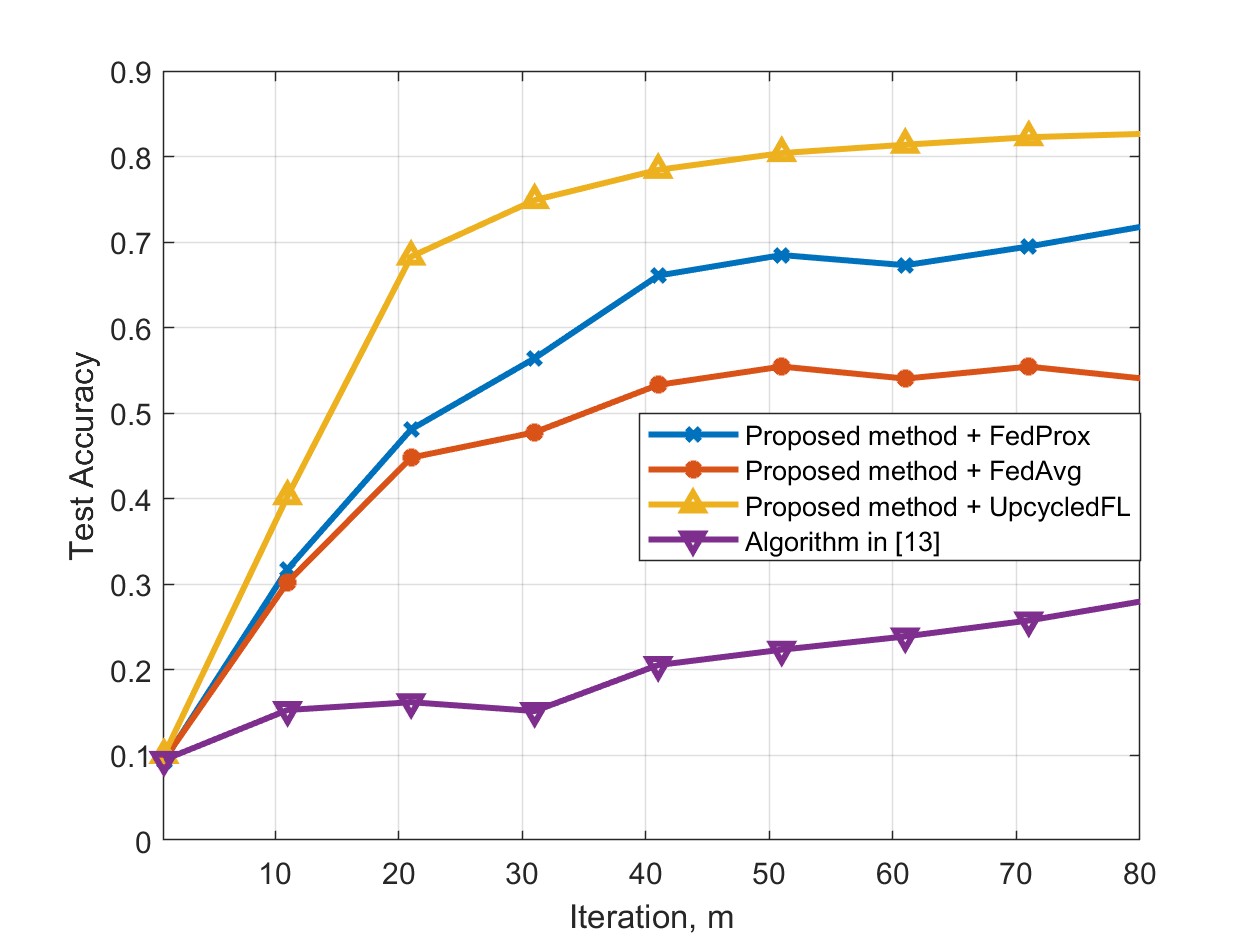}
    \caption{Test accuracy without cooperative jammer $(\varepsilon = 4.4, \delta = 0.01)$.}
    \label{fig:testacc_2}
\end{figure}
\begin{figure}[t]
    \centering
    \includegraphics[width=0.44\textwidth]{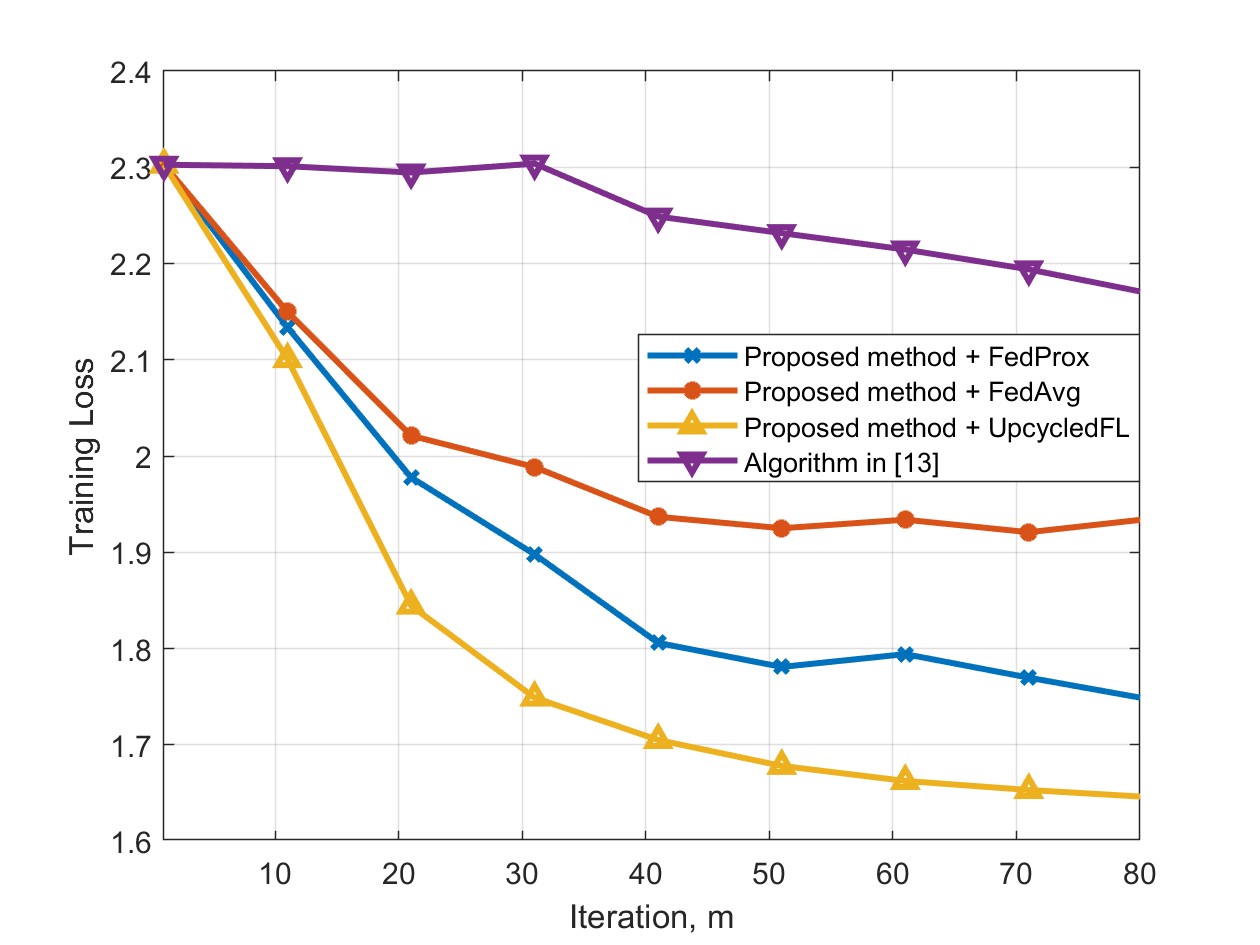}
    \caption{Training loss without cooperative jammer $(\varepsilon = 4.4, \delta = 0.01)$.}
    \label{fig:trainingloss_2}
\end{figure}

We perform image classification tasks using the FEMNIST dataset, a federated adaptation of the EMNIST dataset~\cite{cohen2017emnist}. FEMNIST comprises 62 classes, with each image having $28 \times 28$ pixels. From the EMNIST dataset, we select a subset consisting of 10 lowercase characters, ranging from `a' to `j'. These classes are distributed such that each client receives five, establishing a non-i.i.d. data distribution.
The experiments are conducted in 10 independent trials and we report the average results. For the classification tasks, we use a Multilayer Perceptron (MLP) employing a two-layer network architecture. This network includes a hidden layer with dimensions of $14 \times 14$.
We use a standard cross-entropy as loss function.

In our implementation of Algorithm~\ref{alg:upcycled-ota}, we employ SGD as the local optimizer, utilizing a momentum of 0.5. We set each training session to run for 20 local epochs, with a learning rate of $\eta=0.05$ and a proximal term of $\mu=0.1$. For a fair comparison, we also set the same proximal term, $\mu=0.1$, for FedProx. We adjust hyperparameters for the remaining baselines to achieve optimal performance.
In the Upcycled-FL framework, since even iterations are computed by the PS, we can double the number of iterations relative to other FL frameworks such as FedAvg and FedProx. Consequently, we execute 160 epochs for our Upcycled-FL method, while the baselines are run for 80 epochs. The parameter $\lambda_m$ of Upcycled-FL is set across different segments of the training as follows: it is set to 0.15 for $m \in [1,25]$, increases to 0.4 for $m \in [26,50]$, rises to 0.9 for $m \in [51,75]$, and peaks at 1.9 for $m \in [76,80]$.

We start by comparing the performance of Algorithm~\ref{alg:upcycled-ota}, and also those that incorporate our proposed power control design into FedAvg and FedProx  with the state-of-the-art in~\cite{liu2020privacy}.
To ensure a fair comparison, we exclude the cooperative jammer to maintain consistency in the system model.
In our evaluation, using privacy parameters $\delta=10^{-5}$ and $0.01$, our designs achieve privacy level $\varepsilon=6.52$ and $4.4$, respectively, both within reasonable range.
Subsequently, we apply the same DP settings to the algorithm in~\cite{liu2020privacy} to facilitate a comparative evaluation of learning performance, while maintaining consistency across all other parameters.
The performance metrics, illustrated in Fig.~\ref{fig:testacc_1} and Fig.~\ref{fig:trainingloss_1}, show test accuracy and training loss versus global iterations for the DP level $(\varepsilon=6.52, \delta=10^{-5})$. Similarly, Fig.~\ref{fig:testacc_2} and Fig.~\ref{fig:trainingloss_2} show the results for $(\varepsilon=4.4, \delta=0.01)$.
Results reveal that Algorithm~\ref{alg:upcycled-ota} significantly outperforms all, attributed to the integration of our adaptive power control within the Upcycled-FL framework. 
Notably, the algorithm with FedProx performs better than that with FedAvg, reflecting FedProx's enhanced capability to handle heterogeneous data distributions compared to FedAvg. 
Each of the proposed approach outperforms of the state-of-the-art algorithm~\cite{liu2020privacy}, demonstrating that our power control approach is effective.
The results also emphasize our method's efficacy in maintaining stringent DP levels while ensuring robust learning performance.

\begin{figure}[t]
    \centering
    \includegraphics[width=0.44\textwidth]{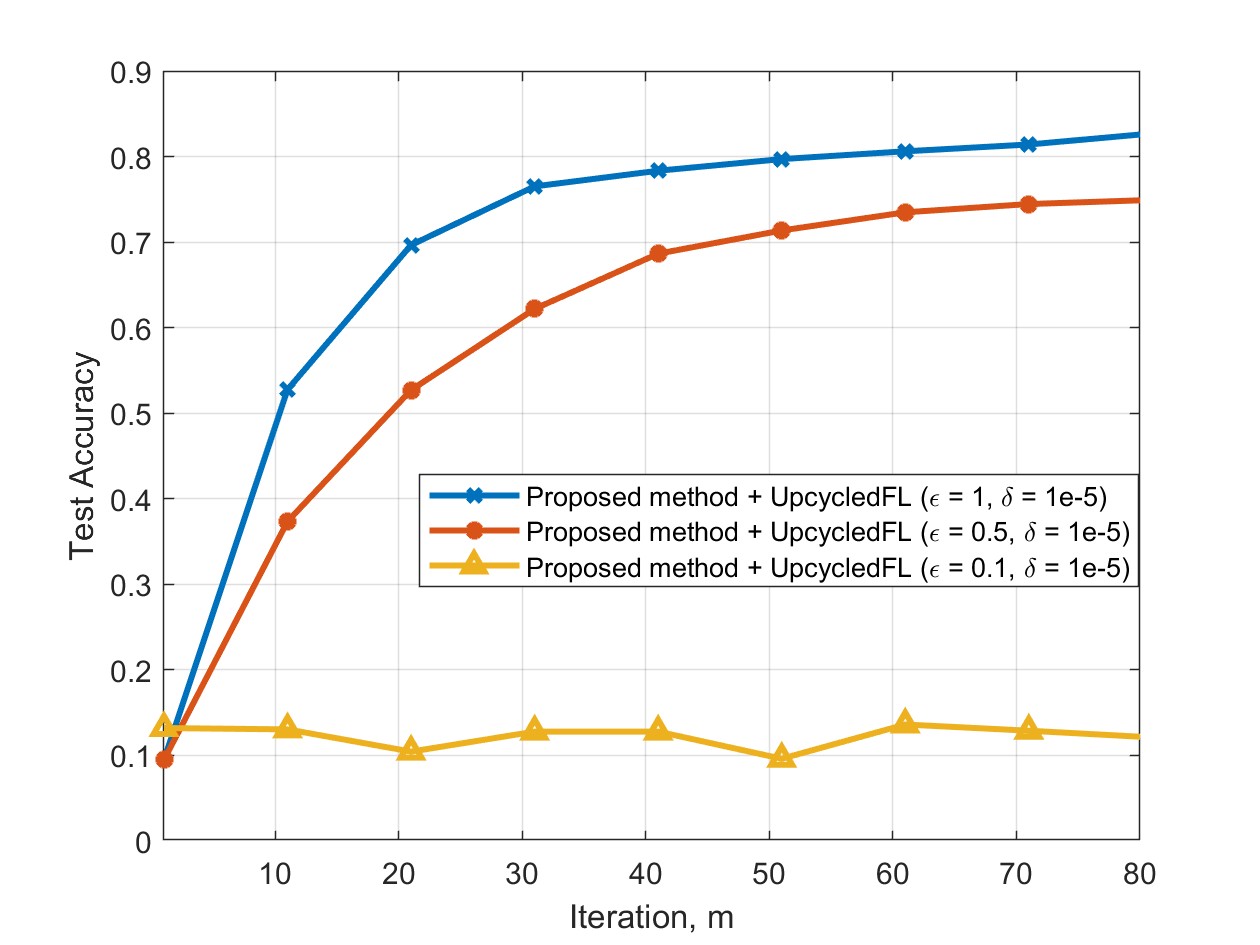}
    \caption{Test accuracy with cooperative jammer for Upcycled-FL.}
    \label{fig:testacc_3}
    \vspace{-4pt}
\end{figure}
\begin{figure}[t]
    \centering
    \includegraphics[width=0.44\textwidth]{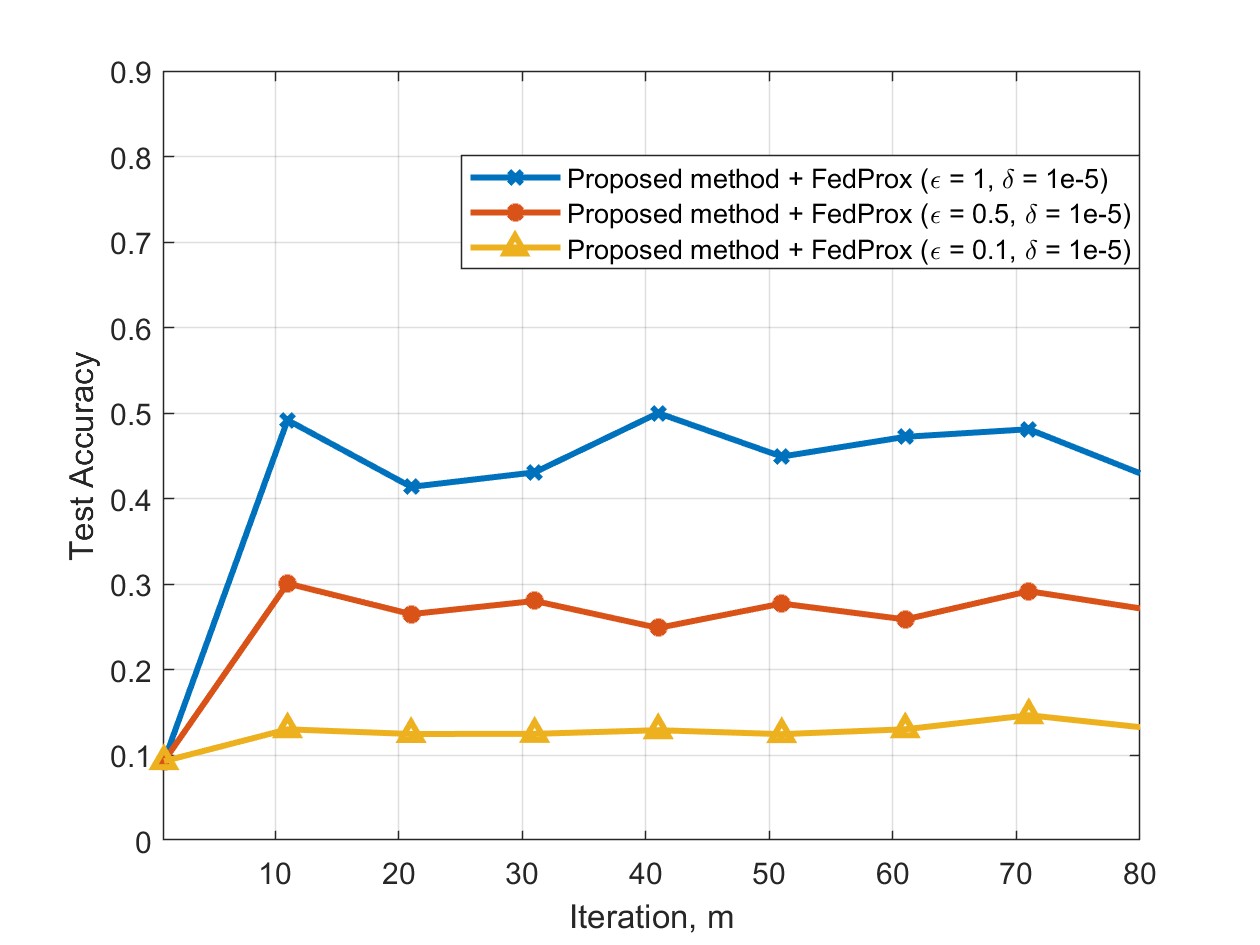}
    \caption{Test accuracy with cooperative jammer for FedProx.}
    \label{fig:testacc_4}
\end{figure}
\begin{figure}[t]
    \centering
    \includegraphics[width=0.44\textwidth]{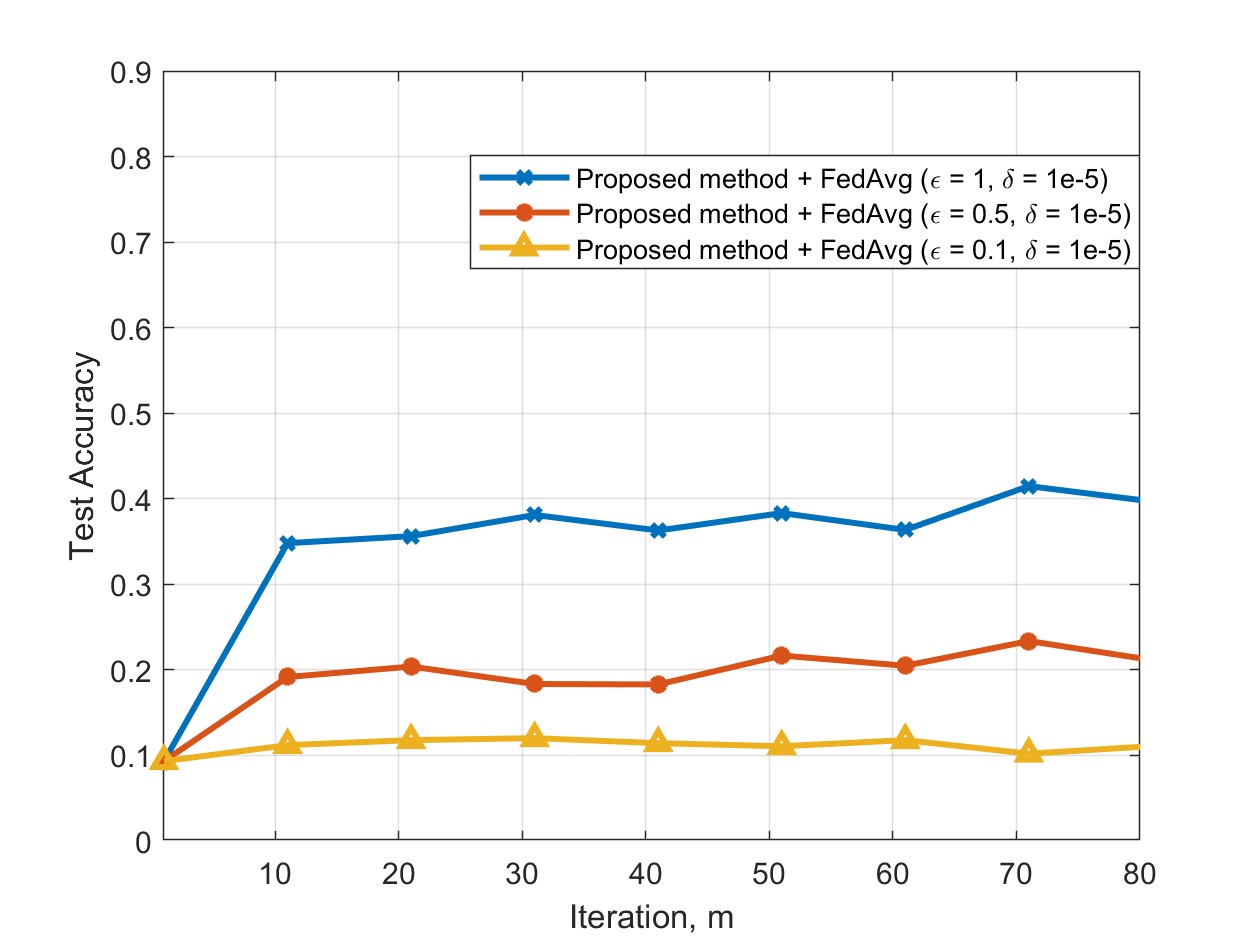}
    \caption{Test accuracy with cooperative jammer for FedAvg.}
    \label{fig:testacc_5}
    \vspace{-5pt}
\end{figure}

We next evaluate the convergence of our design incorporating a cooperative jammer under stricter privacy constraints, i.e., at smaller $\varepsilon$ values. 
We set $\delta$ at $10^{-5}$ and select $\varepsilon$ values of 1, 0.5, and 0.1, all of which are more stringent than the DP levels our algorithm achieves without CJ.
To more accurately depict the observed phenomenon, we apply a factor of 7 times the theoretical lower bound.
Figs.~\ref{fig:testacc_3},~\ref{fig:testacc_4}, and~\ref{fig:testacc_5} show the test accuracy for Upcycled-FL, FedProx, and FedAvg, respectively, across these varying DP settings. 
Consistent with our findings from scenarios without the CJ, Upcycled-FL continues to surpass both FedProx and FedAvg in performance.
It is evident that as $\varepsilon$ decreases, the increased noise adversely affects learning performance. Particularly at very low $\varepsilon$ values, such as 0.1, the excessive noise required to meet stringent privacy levels significantly impedes the algorithms' convergence.
This observation substantiates our theoretical insights into the privacy-accuracy trade-off, where higher privacy levels inherently lead to reduced accuracy.

\subsection{IID Data}
\label{subsec:exp-IID}
\begin{figure}[t]
    \centering
    \includegraphics[width=0.44\textwidth]{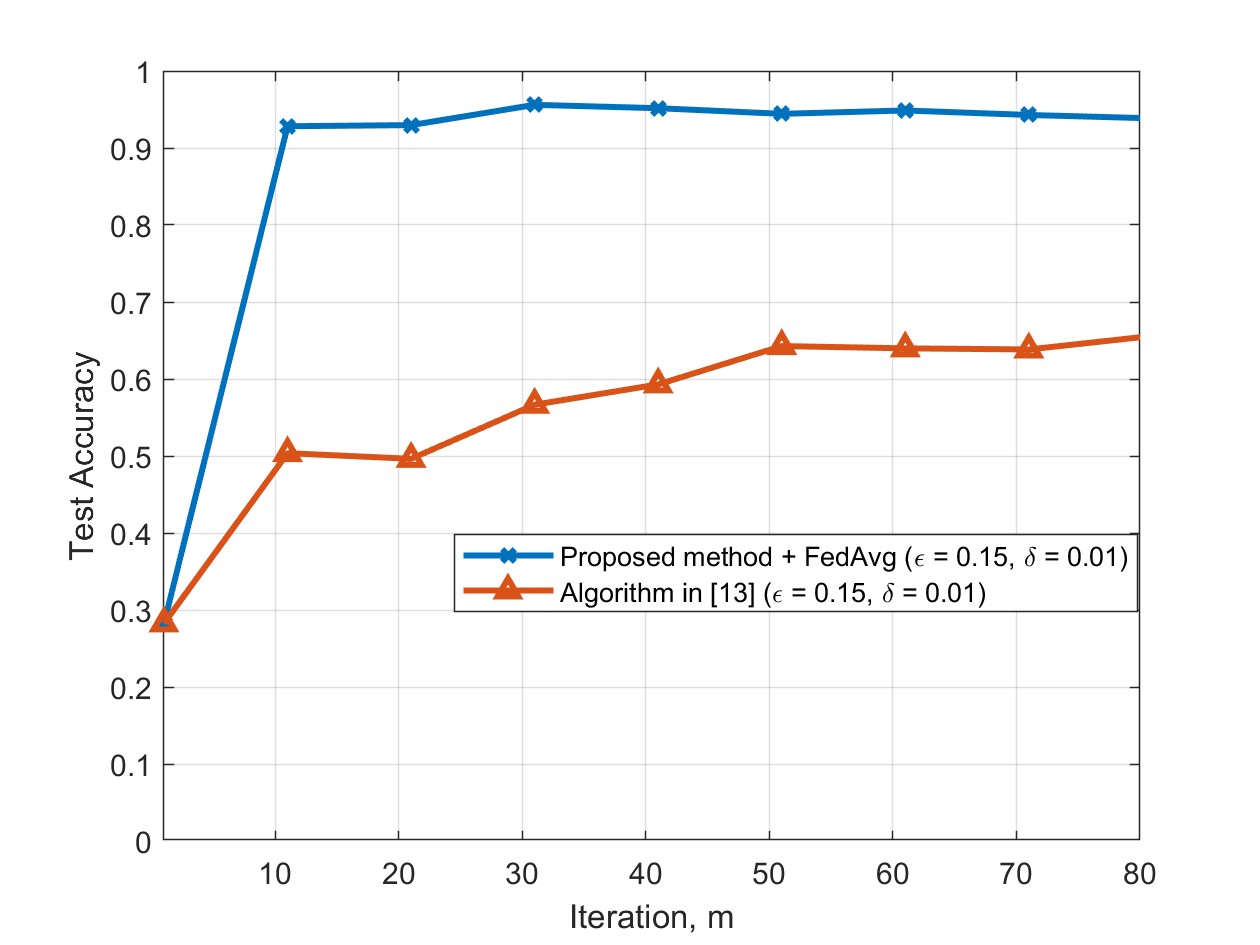}
    \caption{Test accuracy without cooperative jammer $(\varepsilon = 0.15, \delta = 0.01)$.}
    \label{fig:testacc_6}
\end{figure}
\begin{figure}[t]
    \centering
    \includegraphics[width=0.44\textwidth]{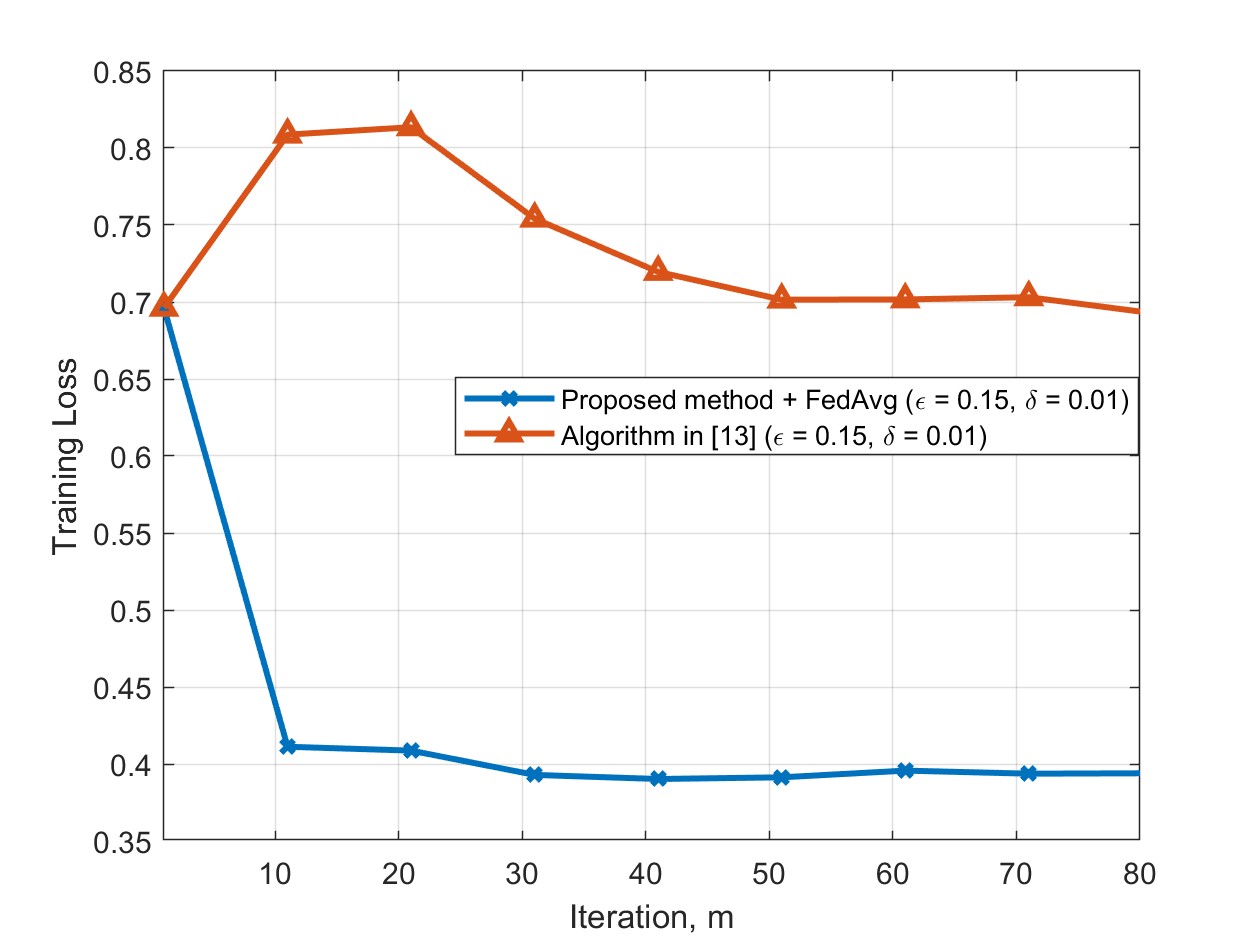}
    \caption{Training Loss without cooperative jammer $(\varepsilon = 0.15, \delta = 0.01)$.}
    \label{fig:trainingloss_3}
\end{figure}

We now shift our focus to an i.i.d data distribution for implementation with FedAvg. 
As previously noted, FedAvg is known to perform poorly with heterogeneous data distributions, prompting us to conduct experiments on i.i.d data to validate the efficacy of our power control strategy within the FedAvg framework. 
Following the approach in~\cite{li2020federated}, we generate a synthetic i.i.d dataset and maintain consistent hyperparameters with the baseline in~\cite{liu2020privacy} to ensure a fair comparison. We set $\delta=0.01$ and achieve a privacy budget $\varepsilon=0.15$ with our method, which we then apply to the baseline.
As illustrated in Fig. \ref{fig:testacc_6} and Fig. \ref{fig:trainingloss_3}, our proposed design significantly outperforms the state-of-the-art algorithm~\cite{liu2020privacy}. This not only demonstrates the effectiveness and robustness of our adaptive power control strategy but also confirms its seamless integration capability with existing FL frameworks.

\section{Conclusion} 
\label{sec: conclusion}
In this paper, we developed a decentralized dynamic power control strategy tailored for differentially private OTA-FL.
This approach utilizes a cooperative jammer as needed to achieve strict privacy requirements without compromising transmission efficiency. 
Our power control method is designed to facilitate easy integration across various FL frameworks.
While it was primarily implemented within the Upcycled-FL framework, we have also extended its application to FedAvg and FedProx.
We employed the Moments Accountant method and derived an artificial noise design for rigorous privacy assessments. 
We provided a convergence analysis on non-convex objectives and performed experiments on the non-i.i.d. real-world dataset FEMNIST.
The numerical results not only demonstrated our method's effectiveness compared to the state-of-the-art under the same DP conditions, but also emphasized the cooperative jammer's benefits in scenarios requiring higher privacy levels.
Future directions include integrating Reconfigurable Intelligent Surfaces~\cite{mao2023roar,mao2023ris, shi2024federated} to further enhance learning performance while adhering to privacy constraints, and considering energy efficiency metrics in providing privacy preserving FL.

\allowdisplaybreaks

\section*{Appendix}
\label{sec:privacyproof}
\subsection{Proof of Theorem 1}
\label{sec:proof_thm_1}
Note: the proof follows the privacy analysis in Upcycled FL, i.e., appendix G proof of theorem C.1.

WLOG, consider the case when local device got updated in every iteration and the algorithm runs over 2M iterations in total. We will use the uppercase letters X and lowercase letters x to denote random variables and the corresponding realizations, and use PX (·) to denote its probability distribution. In appendix G, outputs of mechanism $\mathcal{M}$ are local models of each client. However, in over-the-air computation, we directly have a summation in the server side, the output of such mechanism is the global model update.

Device $i$ transmits signal $\x_i^{2m+1}$ at odd iteration $2m+1$:
\begin{equation}
    \x_i^{2m+1} = \alpha_i^{2m+1} (\w_i^{2m+1} - \Bar{\w}^{2m}) \nonumber
\end{equation}
Device $a$ only transmits artificial noise: $\alpha_{CJ}^{2m+1} \n_{CJ}^{2m+1}$, $\n_{CJ}^{2m+1} \sim \mathcal{N}(0,1)$.

Received signal at the parameter server: 
\begin{equation}
    \y^{2m+1} = \sum_{i \in \mathcal{I}} h^{2m+1}_i \x_i^{2m+1} + h^{2m+1}_{CJ} \alpha_{CJ}^{2m+1} \n_{CJ}^{2m+1} + \z^{2m+1} \nonumber
\end{equation}

Global model update at odd iteration $2m+1$:
\begin{align}
    &\Bar{\w}^{2m+1} - \Bar{\w}^{2m} = \frac{\y^{2m+1}}{\alpha^{2m+1}_u} = \frac{1}{\alpha^{2m+1}_u} \sum_{i \in \mathcal{I}} h^{2m+1}_i \x_i^{2m+1} \nonumber \\
    &+ \frac{1}{\alpha^{2m+1}_u} h^{2m+1}_{CJ} \alpha_{CJ}^{2m+1} \n_{CJ}^{2m+1} + \frac{\z^{2m+1}}{\alpha^{2m+1}_u} \nonumber\\
    &=\frac{1}{\alpha^{2m+1}_u} \sum_{i \in \mathcal{I}} h^{2m+1}_i \alpha_i^{2m+1} (\w_i^{2m+1} - \Bar{\w}^{2m}) \nonumber\\ 
    &+ \frac{h^{2m+1}_{CJ} \alpha_{CJ}^{2m+1}}{\alpha^{2m+1}_u}\n_{CJ}^{2m+1} + \Tilde{\z}^{2m+1}, \nonumber
\end{align}
where $\Tilde{\z}^{2m+1} \sim \mathcal{N}(0,\frac{\sigma_c^2}{(\alpha^{2m+1}_u)^2}\mathbf{I})$, $\sigma_c^2$ is the variance of channel noise.

Still, the privacy loss comes from the odd iterations. Now the mechanism output is $\Bar{\w}^{2m+1}$. For a mechanism $\mathcal{M}$ outputs $o$, with inputs $d$ and $\hat{d}$, let a random variable $c(o; \mathcal{M}, d,\hat{d})=\log \frac{Pr(\mathcal{M}(d)=o)}{Pr(\mathcal{M}(\hat{d})=o)}$ denote the privacy loss at $o$, and 
\begin{equation}
    \alpha_{\mathcal{M}} (\lambda) = \max_{d,\hat{d}} \log \Eb_{o \sim \Mc(d)} \{exp(\lambda c(o; \mathcal{M}, d,\hat{d}))\}. \nonumber
\end{equation}
For client $i$, consider two different datasets $\Dc_i$ and $\Dc_i'$, and datasets of other clients are fixed. The total privacy loss is only contributed by odd iterations. Thus, for the sequence output of global models generated by mechanisms $\{\Mc^m\}_{m=1}^{M}$ over $2M$ iterations, there is:
\begin{align}
    &c(\Bar{\w}^{0:2M};\{\Mc^m\}_{m=1}^{M}, \Dc_i, \Dc_i')  = \log \frac{P_{W^{0:2M}}(\Bar{\w}^{0:2M}|\Dc_i)}{P_{W^{0:2M}}(\Bar{\w}^{0:2M}|\Dc_i')} \nonumber\\
    &= \sum_{m=0}^M \log \frac{P_{W^{2m+1}}(\Bar{\w}^{2m+1}|\Dc_i, \Bar{\w}^{0:2m})}{P_{W^{2m+1}}(\Bar{\w}^{2m+1}|\Dc_i',\Bar{\w}^{0:2m})} + \log \frac{P_{W^0}(\w^0|\Dc_i)}{P_{W^0}(\w^0|\Dc_i')} \nonumber\\
    &= \sum_{m=0}^M c(\Bar{\w}^{2m+1};\Mc^m, \Bar{\w}^{0:2M}, \Dc_i, \Dc_i'),  \nonumber
\end{align}
where $P_{W^0}(\w^0|\Dc_i)=P_{W^0}(\w^0|\Dc_i')$ since $\w^0$ is randomly generated. Moreover, 
\begin{align}
    & \log \Eb_{o \sim \Mc(d)} \{exp(\lambda c(\Bar{\w}^{0:2m}; \{\Mc^m\}_{m=1}^{M}, \Dc_i,\Dc_i'))\} \nonumber \\
    &=  \sum_{m=0}^M \log \Eb_{\w^{2m+1}} \{ exp(\lambda c(\Bar{\w}^{2m+1}; \Mc^m, \Bar{\w}^{0:2m}, \Dc_i, \Dc_i')) \}, \nonumber
\end{align}
and $\alpha_{\{\Mc^m\}_{m=1}^M}(\lambda) \leq \sum_{m=1}^M \alpha_{\Mc^m}(\lambda)$ holds. We first bound $\alpha_{\Mc^m}(\lambda)$. 
\begin{align}
    &\Mc^m(\Dc_i) = \Bar{\w}^{2m+1} = \Bar{\w}^{2m} + \frac{h_i^{2m+1}\alpha_i^{2m+1}}{\alpha_u^{2m+1}} (\w_i^{2m+1} - \Bar{\w}^{2m}) \nonumber \\
    &+ \sum_{j \in \Ic, j\neq i} \frac{h_j^{2m+1} \alpha_j^{2m+1}}{\alpha_u^{2m+1}} (\w_j^{2m+1} - \Bar{\w}^{2m})  + \Tilde{\n}^{2m+1} \nonumber\\
    &=\Bar{\w}^{2m} + \sum_{j \in \Ic, j\neq i} \frac{h_j^{2m+1} \alpha_j^{2m+1}}{\alpha_u^{2m+1}} (\w_j^{2m+1} - \Bar{\w}^{2m}) \nonumber\\
    &+ \frac{h_i^{2m+1}\alpha_i^{2m+1}}{\alpha_u^{2m+1}} \frac{1}{|\Dc_i|} \sum_{d \in \Dc_i} \eta(d) + \Tilde{\n}^{2m+1}, \nonumber
\end{align}
where $\|\eta(\cdot)\|\leq \tau$ and $\Tilde{\n}^{2m+1}$ is effective Gaussian noise, $\Tilde{\n}^{2m+1} \sim \Nc(0, \sigma_{2m+1}^2 \I)$, $\sigma_{2m+1}^2= (\frac{\alpha_{CJ}^{2m+1} h_{CJ}^{2m+1}}{\alpha_u^{2m+1}})^2+\frac{\sigma_c^2}{(\alpha^{2m+1}_u)^2}$.
Without loss of generality, let $\Dc_i' = \Dc_i \cup \{d_n\}$, $\eta(d_n) = \pm \tau \e_1$ and $\sum_{d \in \Dc_i} \eta(d) = \mathbf{0}$. Consider same output of mechanism $\Mc^m$, $(\w_j^{2m+1} - \Bar{\w}^{2m})$ from other clients are same (their datasets are fixed), then the difference of $\Mc^m(\Dc_i)$ and $\Mc^m(\Dc_i')$ comes from one coordinate and the problem can be reduced to one-dimensional problem:
\begin{align}
    &\frac{h_i^{2m+1}\alpha_i^{2m+1}}{\alpha_u^{2m+1}} \frac{1}{|\Dc_i|} \sum_{d \in \Dc_i} \eta(d) + \Tilde{\n}^{2m+1} \nonumber\\
    &= \frac{h_i^{2m+1}\alpha_i^{2m+1}}{\alpha_u^{2m+1}} \frac{1}{|\Dc_i|} \sum_{d \in \Dc_i'} \eta(d) + \Tilde{\n'}^{2m+1} \nonumber\\
    & \Rightarrow \Tilde{\n}^{2m+1} = \pm \frac{h_i^{2m+1}\alpha_i^{2m+1}}{\alpha_u^{2m+1}} \frac{\tau}{|\Dc_i|} \e_1 + \Tilde{\n'}^{2m+1}. \nonumber
\end{align}
\begin{align}
    &c(\Bar{\w}^{2m+1}; \Mc^m, \Bar{\w}^{0:2m}, \Dc_i, \Dc_i') \nonumber\\
    &= \log \frac{P_{W^{2m+1}}(\Bar{\w}^{2m+1}|\Dc_i, \Bar{\w}^{0:2m})}{P_{W^{2m+1}}(\Bar{\w}^{2m+1}|\Dc_i',\Bar{\w}^{0:2m})} \nonumber\\
    &= \log \frac{P_N(\Tilde{n}^{2m+1})}{P_N(\Tilde{n}^{2m+1} \pm \frac{h_i^{2m+1}\alpha_i^{2m+1}}{\alpha_u^{2m+1}}\frac{\tau}{|\Dc_i|})} \nonumber\\
    & \leq \frac{\frac{h_i^{2m+1}\alpha_i^{2m+1}}{\alpha_u^{2m+1}} \tau}{2 |\Dc_i| \sigma_{2m+1}^2} (2|\Tilde{n}^{2m+1}|+ \frac{h_i^{2m+1}\alpha_i^{2m+1}}{\alpha_u^{2m+1}} \tau). \nonumber
\end{align}
Therefore,
\begin{equation}
    \alpha_{\Mc^m}(\lambda) = \frac{(\frac{h_i^{2m+1}\alpha_i^{2m+1}}{\alpha_u^{2m+1}} \tau)^2 \lambda (\lambda+1)}{2 |D_i|^2 \sigma_{2m+1}^2}, \nonumber
\end{equation}
\begin{equation}
    \alpha_{\{\Mc^m\}_{m=1}^M}(\lambda) \leq \sum_{m=1}^M \alpha_{\Mc^m}(\lambda) = \sum_{m=1}^M \frac{(\frac{h_i^{2m+1}\alpha_i^{2m+1}}{\alpha_u^{2m+1}} \tau)^2 \lambda (\lambda+1)}{2 |D_i|^2 \sigma_{2m+1}^2}. \nonumber
\end{equation}


Now we substitute the power control parameters in \eqref{equ: PC_client}. For each client $i$, set $\alpha_i$ as:
\begin{equation}
    \alpha_i^{2m+1} = \frac{\alpha_u^{2m+1} p_i }{h_i^{2m+1} \tau s_i^{2m+1}}, \nonumber
\end{equation}
where $p_i = \frac{|\Dc_i|}{\sum_{j \in \Ic}|\Dc_j|}$ is the proportion of client $i$, $s_i^{2m+1}$ is a hyper-parameter to satisfy transmit power constraint $\Eb \|\x_i^{2m+1}\| \leq P_i$. Specifically, we can find it in each odd iteration as following
\begin{align}
    &\Eb \|\x_i^{2m+1}\|_2^2 = \Eb \| \alpha_i^{2m-1} (\w_i^{2m-1} - \Bar{\w}^{2m})\|_2^2 \nonumber\\
    &= \left\|\frac{\alpha_u^{2m+1} p_i }{h_i^{2m+1} \tau s_i^{2m+1}}\right\|_2^2 \Eb \|\nabla g_i(\w_i^{2m-1})\|_2^2 \nonumber\\
    &\leq \left\|\frac{\alpha_u^{2m+1} p_i }{h_i^{2m+1} \tau s_i^{2m+1}}\right\|_2^2 * \tau^2 =\left\|\frac{\alpha_u^{2m+1} p_i }{h_i^{2m+1}s_i^{2m+1}}\right\|_2^2 \leq P_i. \nonumber
\end{align}
where we assume $\Eb \|\nabla g_i(\w_i^{2m-1})\|_2^2 \leq \tau$, which can be achieved by clipping the update in practice. Then use the tail bound in Theorem 2 \cite{abadi2016deep}, for any $\delta \in [0,1]$, the algorithm is $(\varepsilon,\delta)$-differential private for 
\begin{align}
    \varepsilon & = \min_{\lambda:\lambda \geq0} h_1(\lambda) = \min_{\lambda:\lambda \geq0} \frac{1}{\lambda}\alpha_{\Mc^m}(\lambda) + \frac{1}{\lambda} \log \left(\frac{1}{\delta}\right) \nonumber\\
    &= \min_{\lambda:\lambda \geq0} \sum_{m=1}^M \frac{(\frac{h_i^{2m+1}\alpha_i^{2m+1}}{\alpha_u^{2m+1}} \tau)^2 (\lambda+1)}{2 |D_i|^2 \sigma_{2m+1}^2} + \frac{1}{\lambda} \log \left(\frac{1}{\delta} \right). \nonumber
\end{align}
Take derivative of $h_1(\lambda)$ and set it 0, we can get

\begin{align}
    &\varepsilon = 2 \sqrt{\sum_{m=1}^M \frac{(\frac{h_i^{2m+1}\alpha_i^{2m+1}}{\alpha_u^{2m+1}} \tau)^2}{2 |D_i|^2 \sigma_{2m+1}^2}\log \left(\frac{1}{\delta}\right)} + \sum_{m=1}^M \frac{(\frac{h_i^{2m+1}\alpha_i^{2m+1}}{\alpha_u^{2m+1}} \tau)^2}{2 |D_i|^2 \sigma_{2m+1}^2} \nonumber \\
    & = 2 \sqrt{ \frac{p_i^2}{2|\Dc_i|^2} \sum_{m=1}^M \frac{1}{(s_i^{2m+1})^2\sigma_{2m+1}^2} \log \left(\frac{1}{\delta}\right)} \nonumber \\
    &+ \frac{p_i^2}{2|\Dc_i|^2} \sum_{m=1}^M \frac{1}{(s_i^{2m+1})^2\sigma_{2m+1}^2} \nonumber \\
    & \leq 2 \sqrt{ \frac{p_i^2}{2|\Dc_i|^2} \sum_{m=1}^M \frac{1}{\sigma_{2m+1}^2} \log \left(\frac{1}{\delta}\right)} + \frac{p_i^2}{2|\Dc_i|^2} \sum_{m=1}^M \frac{1}{\sigma_{2m+1}^2} \nonumber \\
    & = 2 \sqrt{ \frac{1}{2|\Dc|^2} \sum_{m=1}^M \frac{1}{\sigma_{2m+1}^2} \log \left(\frac{1}{\delta}\right)} + \frac{1}{2|\Dc|^2} \sum_{m=1}^M \frac{1}{\sigma_{2m+1}^2}, \nonumber
\end{align}
where $|\Dc| = \sum_{i \in \Ic} |\Dc_i|$ is the total number of training data points.
Note that we assign $s_i^{2m+1} \geq 1$ to ensure that the privacy level $\epsilon$ is bounded while the transmit power constraint is satisfied at the same time.

\subsection{Proof of Theorem 2}
\label{sec:proof_thm_2}
Theorem 1 proves the privacy level with only natural channel noise. In this section, we consider the case that CJ sends the artificial noise to the server to achieve any pair of $(\varepsilon, \delta)$, similar to the Theroem 1 in~\cite{abadi2016deep}.
To ensure the $(\varepsilon, \delta)$-differential privacy for all clients, we make the upper bound~\eqref{inequ:privacy_bound} in Theorem~\ref{thm:privacy} be smaller than the required $\varepsilon$. Define an auxiliary variable 
\begin{equation}
    x = \sqrt{ \frac{1}{2|\Dc|^2} \sum_{m=1}^M \frac{1}{(\sigma^{2m+1})^2} \log \left(\frac{1}{\delta}\right)} \nonumber
\end{equation}
Then the goal is to find the solution for inequality
\begin{align}
    & 2x + \frac{x^2}{\log \left(\frac{1}{\delta}\right)} \leq \varepsilon \nonumber\\
    \Rightarrow \quad & x^2 + 2 \log \left(\frac{1}{\delta}\right) x - \varepsilon \log \left(\frac{1}{\delta}\right) \leq 0 \nonumber\\
    \Rightarrow \quad & 0 < x \leq - \log \left(\frac{1}{\delta}\right) +  \sqrt{\left(\log \left(\frac{1}{\delta}\right) \right)^2 + \varepsilon \log \left(\frac{1}{\delta}\right)}.  \nonumber
\end{align}
Define
\begin{equation}
    a = - \log \left(\frac{1}{\delta}\right) +  \sqrt{\left(\log \left(\frac{1}{\delta}\right) \right)^2 + \varepsilon \log \left(\frac{1}{\delta}\right)}. \nonumber
\end{equation}
Then $x^2 \leq a^2$, let $(\sigma^{2m+1})^2 = \sigma^2$,
\begin{align}
    &\frac{1}{2|\Dc|^2} \sum_{m=1}^M \frac{1}{(\sigma^{2m+1})^2} \log \left(\frac{1}{\delta}\right) \leq a^2, \nonumber\\
    & \sigma^2 \geq \frac{M}{2|\Dc|^2 a^2} \log \left(\frac{1}{\delta}\right) \Rightarrow \nonumber\\
    & \left( \frac{\alpha_{CJ}^{2m+1}h_{CJ}^{2m+1}}{\alpha_u^{2m+1}}\right)^2 = \sigma^2 - \frac{\sigma_c^2}{\alpha_u^{2m+1}} \nonumber\\
    & \geq \frac{M}{2|\Dc|^2 a^2} \log \left(\frac{1}{\delta}\right) - \frac{\sigma_c^2}{\alpha_u^{2m+1}}.  \nonumber
\end{align}
Now we get the design
\begin{equation}
    \alpha^{2m+1}_{CJ} \geq \frac{\alpha^{2m+1}_u}{h^{2m+1}_{CJ}} \sqrt{ \frac{M}{2|\Dc|^2 a^2} \log \frac{1}{\delta} - \frac{\sigma_{c}^2}{(\alpha^{2m+1}_u)^2}}. \nonumber
\end{equation}

\subsection{Proof of Theorem 3}
\label{sec:convgproof}
Update rule of learning process:
\begin{align}
    &\w_i^{2m-1} = \Bar{\w}^{2m-2} - \frac{1}{\mu} \nabla F_i(\w_i^{2m-1}), \nonumber\\
    &\Bar{\w}^{2m-1} = \Bar{\w}^{2m-2} + \frac{1}{\alpha_u} \sum_{i \in \Ic} h_i^{2m-1} \alpha_i^{2m-1} (\w_i^{2m-1} -\Bar{\w}^{2m-2}) \nonumber\\
    &+ \frac{h_{CJ}^{2m-1} \alpha_{CJ}^{2m-1}}{\alpha_u} \n_{CJ}^{2m-1} + \Tilde{\z}^{2m-1}, \nonumber\\
    &\Bar{\w}^{2m} = \Bar{\w}^{2m-1} + \frac{\mu}{\mu+\lambda_m} (\Bar{\w}^{2m-1} - \Bar{\w}^{2m-2}), \nonumber \\
    &\Bar{\w}^{2m+1} - \Bar{\w}^{2m-1} = \frac{\mu}{\mu+\lambda_m} (\Bar{\w}^{2m-1} - \Bar{\w}^{2m-2})  + \Tilde{\z}^{2m-1}  \nonumber \\
    &   - \frac{1}{\mu} \sum_{i \in \Ic} \frac{h_i^{2m-1} \alpha_i^{2m-1}}{\alpha_u} \nabla F_i (\w_i^{2m+1}) + \frac{h_{CJ}^{2m-1} \alpha_{CJ}^{2m-1}}{\alpha_u} \n_{CJ}^{2m-1}. \nonumber
\end{align}
Take expectation conditioned on $\w^{2m-1}$, the randomness is from the channel noise and artificial noise.
\begin{align}
    &\Eb[f(\Bar{\w}^{2m+1})] \leq f(\Bar{\w}^{2m-1}) + \frac{L}{2} \Eb \|\Bar{\w}^{2m+1} - \Bar{\w}^{2m-1}\|^2  \nonumber\\
    &+ \Eb <\nabla f(\Bar{\w}^{2m-1}), \Bar{\w}^{2m+1} - \Bar{\w}^{2m-1}> \nonumber\\
    =& f(\Bar{\w}^{2m-1}) + \Eb <\nabla f(\Bar{\w}^{2m-1}),\frac{\mu}{\mu+\lambda_m} (\Bar{\w}^{2m-1} - \Bar{\w}^{2m-2}) \nonumber \\
    &- \frac{1}{\mu} \sum_{i \in \Ic}\frac{h_i^{2m+1} \alpha_i^{2m+1}}{\alpha_u}\nabla F_i (\w_i^{2m+1}) \nonumber \\
    &+\frac{h_{CJ}^{2m+1} \alpha_{CJ}^{2m+1}}{\alpha_u} \n_{CJ}^{2m+1} + \Tilde{\z}^{2m+1}> \nonumber\\
    & + \frac{L}{2} \Eb \|\frac{\mu}{\mu+\lambda_m} (\Bar{\w}^{2m-1} - \Bar{\w}^{2m-2}) + \Tilde{\z}^{2m+1} \nonumber \\
    &- \frac{1}{\mu} \sum_{i \in \Ic}\frac{h_i^{2m+1} \alpha_i^{2m+1}}{\alpha_u}\nabla F_i (\w_i^{2m+1}) +\frac{h_{CJ}^{2m+1} \alpha_{CJ}^{2m+1}}{\alpha_u} \n_{CJ}^{2m+1}  \|^2 \nonumber\\
    =& f(\Bar{\w}^{2m-1}) + <\nabla f(\Bar{\w}^{2m-1}),\frac{\mu}{\mu+\lambda_m} (\Bar{\w}^{2m-1} - \Bar{\w}^{2m-2})> \nonumber\\
    &+ <\nabla f(\Bar{\w}^{2m-1}), - \frac{1}{\mu} \sum_{i \in \Ic} \frac{p_i}{\tau s_i^{2m+1}} \nabla F_i(\w_i^{2m+1})> \nonumber\\
    &+ \frac{L}{2} \frac{\sigma_c^2 d}{\alpha_u^2} + \frac{L}{2} \Eb \|\frac{\mu}{\mu+\lambda_m} (\Bar{\w}^{2m-1} - \Bar{\w}^{2m-2}) \nonumber \\
    &- \frac{1}{\mu} \sum_{i \in \Ic} \frac{p_i}{\tau s_i^{2m+1}} \nabla F_i(\w_i^{2m+1}) \|^2 + \frac{Ld}{2} \left( \frac{h_{CJ}^{2m+1} \alpha_{CJ}^{2m+1}}{\alpha_u}\right)^2  \nonumber\\
    \leq & f(\Bar{\w}^{2m-1}) + \frac{\mu}{\mu+\lambda_m} \|\nabla f(\w^{2m-1})\| \|\Bar{\w}^{2m-1}-\Bar{\w}^{2m-2}\| \nonumber \\
    &- \frac{1}{\mu} <\nabla f(\w^{2m-1}),\sum_{i \in \Ic} \frac{p_i}{\tau s_i^{2m+1}} \nabla F_i(\w_i^{2m+1})> \nonumber\\
    & + L \left(\frac{\mu}{\mu+\lambda_m}\right)^2 \|\Bar{\w}^{2m-1}-\Bar{\w}^{2m-2}\|^2 + \frac{L}{2} \frac{\sigma_c^2 d}{\alpha_u^2} \nonumber \\
    &+ \frac{L}{\mu^2} \left\|\sum_{i \in \Ic} \frac{p_i}{\tau s_i^{2m+1}} \nabla F_i(\w_i^{2m+1})\right\|^2 +\frac{Ld}{2} \left( \frac{h_{CJ}^{2m+1} \alpha_{CJ}^{2m+1}}{\alpha_u}\right)^2.    \nonumber
\end{align}
Note that 
\begin{align}
    &- \frac{1}{\mu} <\nabla f(\w^{2m-1}),\sum_{i \in \Ic} \frac{p_i}{\tau s_i^{2m+1}} \nabla F_i(\w_i^{2m+1})>  \nonumber\\
    &= \frac{1}{2\mu} \bigg( \left\| \nabla f(\w^{2m-1}) - \sum_{i \in \Ic} \frac{p_i}{\tau s_i^{2m+1}} \nabla F_i(\w_i^{2m+1}) \right\|^2  \nonumber\\
    &- \|\nabla f(\w^{2m-1})\|^2- \left\| \sum_{i \in \Ic}\frac{p_i}{\tau s_i^{2m+1}} \nabla F_i(\w_i^{2m+1}) \right\|^2\bigg) \nonumber
\end{align}
Thus,
\begin{align}
    &\Eb[f(\Bar{\w}^{2m+1})] \leq f(\Bar{\w}^{2m-1})+ L \frac{\mu^2}{(\mu+\lambda_m)^2}\|\Bar{\w}^{2m-1}-\Bar{\w}^{2m-2}\|^2  \nonumber \\
    & + \frac{\mu}{\mu+\lambda_m} \|\nabla f(\w^{2m-1})\| \|\Bar{\w}^{2m-1}-\Bar{\w}^{2m-2}\|\nonumber\\
    & + \frac{Ld}{2} \left( \frac{h_{CJ}^{2m+1} \alpha_{CJ}^{2m+1}}{\alpha_u}\right)^2 + \frac{L}{2} \frac{\sigma_c^2 d}{\alpha_u^2} - \frac{1}{2 \mu} \|\nabla f(\Bar{\w}^{2m-1}) \|^2 \nonumber\\
    & + \frac{1}{2\mu} \left\| \nabla f(\w^{2m-1}) - \sum_{i \in \Ic} \frac{p_i}{\tau s_i^{2m+1}} \nabla F_i(\w_i^{2m+1}) \right\|^2  \nonumber\\
    &+ \left(\frac{L}{\mu^2} - \frac{1}{2 \mu} \right) \left\| \sum_{i \in \Ic}\frac{p_i}{\tau s_i^{2m+1}} \nabla F_i(\w_i^{2m+1}) \right\|^2 \label{inequ: partP}
\end{align}
Next, we will bound the last two terms. By Jensen's inequality and $F_i$ is $L$-smooth,
\begin{align}
    &\frac{1}{2\mu} \left\| \nabla f(\Bar{\w}^{2m-1}) - \sum_{i \in \Ic} \frac{p_i}{\tau s_i^{2m+1}} \nabla F_i(\w_i^{2m+1}) \right\|^2 \nonumber\\
    =& \|\sum_{i \in \Ic} p_i \nabla F_i(\Bar{\w}^{2m-1}) - \sum_{i \in \Ic} p_i \nabla F_i(\w_i^{2m+1}) \nonumber \\
    &+\sum_{i \in \Ic} p_i \nabla F_i(\w_i^{2m+1}) - \sum_{i \in \Ic} \frac{p_i}{\tau s_i^{2m+1}} \nabla F_i(\w_i^{2m+1}) \|^2 \nonumber\\
    =&  \frac{1}{2\mu} \| \Eb_i[\nabla F_i(\Bar{\w}^{2m-1}) - \nabla F_i(\w_i^{2m+1})] \nonumber \\
    &+ \Eb_i\left(1 - \frac{1}{\tau s_i^{2m+1}} \right) \nabla F_i(\w_i^{2m+1}) \|^2 \nonumber\\
    \leq & \frac{1}{\mu} \|\Eb_i[\nabla F_i(\Bar{\w}^{2m-1}) - \nabla F_i(\w_i^{2m+1})]\|^2 \nonumber\\
    &+ \frac{1}{\mu} \left\|\Eb_i\left(1 - \frac{1}{\tau s_i^{2m+1}} \right) \nabla F_i(\w_i^{2m+1})\right\|^2 \nonumber \\
    \leq & \frac{1}{\mu} \Eb_i\|\nabla F_i(\Bar{\w}^{2m-1}) - \nabla F_i(\w_i^{2m+1})\|^2 \nonumber \\
    &+ \frac{1}{\mu} \Eb_i\left\|\left(1 - \frac{1}{\tau s_i^{2m+1}} \right) \nabla F_i(\w_i^{2m+1})\right\|^2 \nonumber \\
    \leq & \frac{L}{\mu} \Eb_i \|\w_i^{2m+1}-\Bar{\w}^{2m-1}\|^2 \nonumber \\
    &+ \frac{1}{\mu}\Eb_i (1 - \frac{1}{\tau s_i^{2m+1}})^2 \|\nabla F_i(\w_i^{2m+1})\|^2. \label{inequ: partP_2}
\end{align}
Apply Jensen's inequality to \eqref{inequ: partP},
\begin{align}
   &\left(\frac{L}{\mu^2} - \frac{1}{2 \mu} \right) \left\| \sum_{i \in \Ic}\frac{p_i}{\tau s_i^{2m+1}} \nabla F_i(\w_i^{2m+1}) \right\|^2 \nonumber\\
    \leq & \left(\frac{L}{\mu^2} - \frac{1}{2 \mu} \right) \Eb_i \left(\frac{1}{\tau s_i^{2m+1}}\right)^2 \|\nabla F_i(\w_i^{2m+1})\|^2. \label{inequ: partP_3}
\end{align}
Apply Assumption \ref{a_bounds} and Lemma \ref{lemma_boundg}, we have
\begin{align}
    \|\nabla F_i(\w_i^{2m+1}) \|^2 \leq (\kappa_i + \|\nabla f(\w_i^{2m+1})\|)^2 \leq (\kappa_i + G)^2. \nonumber
\end{align}
Due to Assumption~\ref{a_smooth} and~\ref{a_strongconv}, we have
\begin{align}
    &\|\w_i^{2m+1}-\Bar{\w}^{2m-1}\| \leq \|\w_i^{2m+1}-\Bar{\w}^{2m}\| + \|\Bar{\w}^{2m} - \Bar{\w}^{2m-1}\| \nonumber\\
    & \leq \frac{L + \rho}{\rho} \|\Bar{\w}^{2m} - \Bar{\w}^{2m-1}\| + \frac{1}{\rho} \|\nabla F_i(\Bar{\w}^{2m-1}) \| \nonumber
\end{align}
Then continue to bound \eqref{inequ: partP_2}, \eqref{inequ: partP_3}
\begin{align}
    &\frac{1}{2\mu} \left\| \nabla f(\w^{2m-1}) - \sum_{i \in \Ic} \frac{p_i}{\tau s_i^{2m+1}} \nabla F_i(\w_i^{2m+1}) \right\|^2 \nonumber \\
    &+ \left(\frac{L}{\mu^2} - \frac{1}{2 \mu} \right) \left\| \sum_{i \in \Ic}\frac{p_i}{\tau s_i^{2m+1}} \nabla F_i(\w_i^{2m+1}) \right\|^2 \nonumber\\
    \leq & \frac{L}{\mu} \Eb_i \left(\frac{L+\rho}{\rho} \|\Bar{\w}^{2m} - \Bar{\w}^{2m-1}\| + \frac{1}{\rho} \|\nabla F_i(\Bar{\w}^{2m-1}) \|\right)^2 \nonumber \\
    &+ \Eb_i [\bigg( \frac{2\mu(\tau s_i^{2m+1}-1)^2 + (2L-\mu)\mu }{2 \mu^2 (\tau s_i^{2m+1})^2}\bigg) (\kappa_i + d)^2] \nonumber\\
    \leq & \frac{L}{\mu} \left(\frac{L+\rho}{\rho} \right)^2 \|\Bar{\w}^{2m} - \Bar{\w}^{2m-1}\|^2 \nonumber \\
    &+ \frac{2L}{\mu}\Eb_i \left[\frac{L+\rho}{\rho}\|\Bar{\w}^{2m} - \Bar{\w}^{2m-1}\| \frac{1}{\rho} \|\nabla F_i(\Bar{\w}^{2m-1}) \|\right] \nonumber\\
    & + \frac{L}{\mu \rho^2} \Eb_i\|\nabla F_i(\Bar{\w}^{2m-1})\|^2  \nonumber \\
    &+ \Eb_i \left[\left( \frac{2\mu(\tau s_i^{2m+1}-1)^2 + (2L-\mu)\mu }{2 \mu^2 (\tau s_i^{2m+1})^2}\right) (\kappa_i + d)^2\right] \nonumber\\
    \leq& \frac{L}{\mu} \left(\frac{L+\rho}{\rho} \right)^2 \|\Bar{\w}^{2m} - \Bar{\w}^{2m-1}\|^2 + \frac{LB}{\mu \rho^2}\|\nabla f(\w^{2m-1})\|^2\nonumber \\
    &+ \frac{2L}{\mu} \left(\frac{L+\rho}{\rho^2}\right)\|\Bar{\w}^{2m} - \Bar{\w}^{2m-1}\|B \|\nabla f(\w^{2m-1})\| \nonumber\\
    &+ \Eb_i \left[\left( \frac{2\mu(\tau s_i^{2m+1}-1)^2 + (2L-\mu)\mu }{2 \mu^2 (\tau s_i^{2m+1})^2}\right) (\kappa_i + G)^2\right]. \nonumber
\end{align}
Re-organize, we get
\begin{align}
    &\Eb[f(\Bar{\w}^{2m+1})] \leq f(\Bar{\w}^{2m-1}) - \left(\frac{1}{2\mu} - \frac{LB}{\mu \rho^2}\right) \|\nabla f(\w^{2m-1})\|^2 \nonumber \\
    &+ \frac{\mu}{\mu+\lambda_m} (1+\frac{2LB(L+\rho)}{\mu\rho^2})\|\Bar{\w}^{2m} - \Bar{\w}^{2m-1}\|\|\nabla f(\w^{2m-1})\| \nonumber\\
    &+ \left( \frac{\mu}{\mu+\lambda_m}\right)^2 L \left(1+\frac{(L+\rho)^2}{\mu \rho^2}\right) \|\Bar{\w}^{2m} - \Bar{\w}^{2m-1}\|^2 \nonumber\\
    & + \Eb_i \left[\left( \frac{2\mu(\tau s_i^{2m+1}-1)^2 + (2L-\mu)\mu }{2 \mu^2 (\tau s_i^{2m+1})^2}\right) (\kappa_i + G)^2\right] \nonumber\\
    & + \frac{Ld}{2} \left( \frac{h_{CJ}^{2m+1} \alpha_{CJ}^{2m+1}}{\alpha_u}\right)^2 + \frac{L}{2} \frac{\sigma_c^2d}{\alpha_u^2}. \nonumber
\end{align}
Define
\begin{align}
    &\mathbf{C}_1 = \frac{1}{2\mu} - \frac{LB}{\mu \rho^2}, \nonumber\\
    &\mathbf{C}_2^m = \frac{\mu}{\mu+\lambda_m} \left(1+\frac{2LB(L+\rho)}{\mu\rho^2} \right) qG, \nonumber\\
    &\mathbf{C}_3^m = \left( \frac{\mu}{\mu+\lambda_m}\right)^2 L \left(1+\frac{(L+\rho)^2}{\mu \rho^2}\right) q^2, \nonumber\\
    &\mathbf{C}_4^m = \Eb_i \left[\left( \frac{2\mu(\tau s_i^{2m+1}-1)^2 + (2L-\mu)\mu }{2 \mu^2 (\tau s_i^{2m+1})^2}\right) (\kappa_i + G)^2\right], \nonumber\\
    &\mathbf{C}_5^m = \frac{Ld}{2} \left( \frac{h_{CJ}^{2m+1} \alpha_{CJ}^{2m+1}}{\alpha_u}\right)^2, \nonumber\\
    &\mathbf{C}_6 = \frac{L\sigma_c^2d}{2\alpha_u^2}. \nonumber
\end{align}
Average over M odd iterations, we get
\begin{align}
    &\min_{m \in [M]} \mb{E} \| \nabla f(\Bar{\w}^{2m-1}) \|^2 \leq \frac{1}{M} \sum_{m=1}^M \Eb \|\nabla f(\Bar{\w}^{2m-1}) \|^2 \nonumber\\
    & \leq \frac{f(\Bar{\w}^0) - f(\Bar{\w}^*)}{M \mathbf{C_1}} + \frac{\mathbf{C}_6}{ \mathbf{C_1}} + \frac{1}{M \mathbf{C_1}} \sum_{m=1}^M \mathbf{C}_2^m + \mathbf{C}_3^m + \mathbf{C}_4^m + \mathbf{C}_5^m.  \nonumber
\end{align}


\vspace{-0.5in}
\bibliographystyle{IEEEtran}{}
\bibliography{refs_journal}

\end{document}